\documentclass{article} 
\usepackage{conf2027_conference,times}

\usepackage[T1]{fontenc}
\usepackage{amsmath,amssymb,amsthm,mathtools}
\usepackage{graphicx,booktabs,tabularx,array,longtable,wrapfig,capt-of}
\usepackage{tikz}
\usetikzlibrary{arrows.meta,positioning,fit,backgrounds,calc}
\usepackage{microtype}
\usepackage{float}
\usepackage[hidelinks]{hyperref}
\usepackage{url}
\usepackage[normalem]{ulem}
\usepackage{amsmath}
\usepackage{enumitem}
\usepackage[ruled,vlined,linesnumbered]{algorithm2e}
\DontPrintSemicolon
\SetAlgoNlRelativeSize{-1}

\newtheorem{theorem}{Theorem}
\newtheorem{proposition}{Proposition}
\newtheorem{lemma}{Lemma}
\newtheorem{corollary}{Corollary}
\theoremstyle{definition}
\newtheorem{assumption}{Assumption}

\theoremstyle{remark}

\theoremstyle{plain}
\newcommand{\R}{\mathbb{R}}

\newcommand{\cA}{\mathcal{A}}
\newcommand{\cG}{\mathcal{G}}
\newcommand{\cR}{\mathcal{R}}
\newcommand{\cE}{\mathcal{E}}

\newcommand{\alg}{\textup{\textsc{QARS}}}

\DeclareMathOperator{\tr}{tr}

\AtBeginDocument{%
  \setlength{\abovedisplayskip}{4pt plus 2pt minus 3pt}%
  \setlength{\belowdisplayskip}{4pt plus 2pt minus 3pt}%
  \setlength{\abovedisplayshortskip}{0pt plus 2pt}%
  \setlength{\belowdisplayshortskip}{3pt plus 2pt minus 2pt}%
  \setlength{\jot}{2pt}%
  \setlength{\textfloatsep}{12pt plus 2pt minus 4pt}%
  \setlength{\floatsep}{10pt plus 2pt minus 3pt}%
  \setlength{\intextsep}{10pt plus 2pt minus 3pt}%
  \setlength{\abovecaptionskip}{6pt}%
  \setlength{\belowcaptionskip}{0pt}%
}
\makeatletter
\def\thm@space@setup{\thm@preskip=3pt plus1pt minus2pt \thm@postskip=\thm@preskip}
\makeatother

\title{Quality-Aware Cross-Model\\ Computation Reuse}
\author{
Jin Cheng$^{1}$, Xiangxiang Dai$^{1}$, Maoli Liu$^{1}$,\\
\textbf{Ziyi Han$^{1}$, Zhuohua Li$^{2}$, John C.~S.~Lui$^{1}$}\\[0.4em]
$^{1}$The Chinese University of Hong Kong\\
$^{2}$Xi'an University of Electronic Science and Technology
}
\conffinalcopy

\begin{document}
\maketitle
\pagestyle{plain}
\thispagestyle{plain}

\begin{abstract}
An intermediate result computed by one model can be reused by other models to perform their tasks. Existing work mainly focuses on practical execution, leaving a theoretical gap in optimizing reuse decisions. This optimization faces two challenges: \emph{quality uncertainty}, because the effect of reuse on task quality is uncertain across models, and \emph{coupled scheduling}, because tasks need to share the cost of preparing reusable results. These challenges compound each other: quality must be learned online, but the shared preparation structure makes scheduling NP-hard even with known quality, breaking the key assumption in existing methods.

We formulate cross-model computation reuse as an online decision problem and develop the Quality-Aware Reuse Scheduling (\alg{}) algorithm to address it. For \emph{quality uncertainty}, \alg{} learns task-dependent reuse quality from selected, possibly delayed feedback and uses optimistic estimates to guide decisions. For \emph{coupled scheduling}, it jointly chooses which results to prepare and which tasks should use them, adapting scheduling accuracy to the remaining quality uncertainty. For the considered problem, our analysis separates learning and optimization error in regret and quantifies the tradeoff between scheduling accuracy and computation. Completing the quality-aware stopping rule yields $\widetilde O(\sqrt{T})$ regret while preserving feasibility.  Experiments demonstrate the effectiveness of \alg{} in optimizing cross-model reuse, reducing the combined cost of computation and quality loss by up to 18.0\%, and mean regret by 63.9\% over the strongest scheduling baseline.
\end{abstract}

\section{Introduction}\label{sec:intro}
\suppressfloats[t]
Large language models (LLMs) often process the same input for many tasks, such as several questions about one document or successive turns of a conversation. Such repeated processing is common in serving workloads and adds latency and cost \citep{wang2025kvcachewild,ye2025kvcomm}. Computation reuse avoids it by computing a result once and sharing it across tasks, for example by caching the attention states of a shared prefix \citep{gim2024promptcache,zheng2024sglang,yao2025cacheblend}. This reuse typically assumes one model both produces and consumes the result. However, LLM workflows increasingly combine multiple models \citep{zhang2025maas}, so one model's result can also serve another.

\emph{Cross-model computation reuse} targets this case: a \emph{source} model computes a \emph{reusable result}, and a \emph{target} model uses it for multiple tasks (Figure~\ref{fig:reuse-motivation}). For example, a small source model can summarize or compress a long document \citep{jiang2023llmlingua,xu2024recomp}, and a larger target model then answers several questions from this result instead of the full document. Similarly, KV caches can be shared or converted across models \citep{liu2026droidspeak,fu2026cachetocache}. Such reuse avoids repeated processing, but the target model sees a condensed or converted input, which may reduce task quality.

\begin{wrapfigure}{r}{0.37\textwidth}
\vspace{-0.3cm}
\centering
\includegraphics[width=\linewidth]{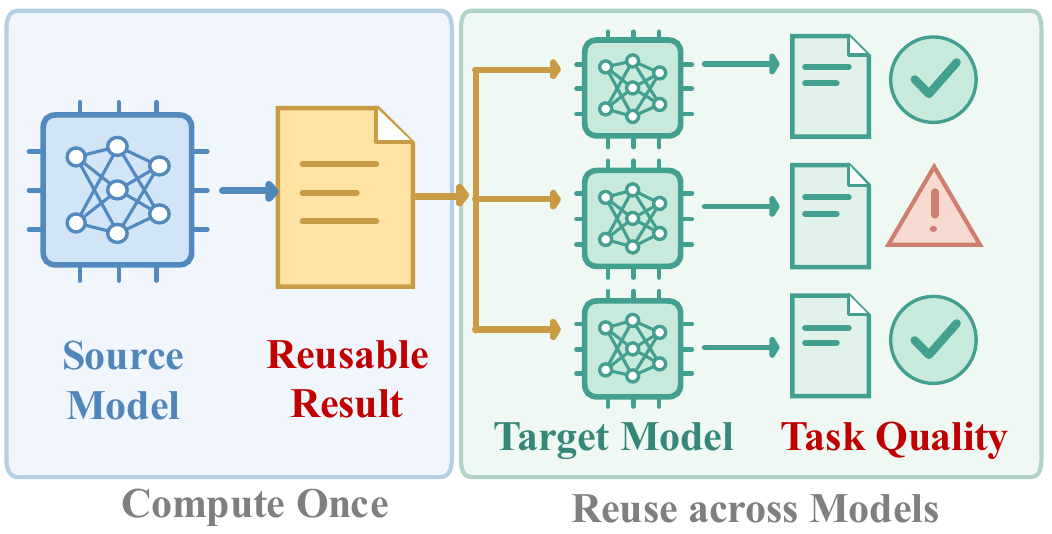}
\caption{Cross-model reuse.}
\label{fig:reuse-motivation}
\vspace{-0.9cm}
\end{wrapfigure}

Existing reuse work mainly focuses on practical execution, leaving a theoretical gap in optimizing reuse decisions. Prior studies develop mechanisms for cross-model reuse \citep{liu2026droidspeak,fu2026cachetocache} and model--cache selection \citep{lu2026gateway}. Building on these systems, we study the reuse scheduling problem: \emph{\uline{which results should be prepared for reuse, and how should tasks be assigned to them when task-dependent quality is unknown and preparation costs are shared?}} Optimizing these decisions faces two coupled challenges.

\emph{\uline{Challenge I: Quality uncertainty}} arises because the effect of reuse on task quality is uncertain across models. Reuse quality depends on the source and target models, the task, and how the result is reused \citep{liu2026droidspeak,fu2026cachetocache,heo2026crossmodel}. In Figure~\ref{fig:reuse-motivation}, a summary that the target model can read may still lack the evidence needed to answer a question. Quality feedback is available only for executed assignments and may be delayed. The scheduler must therefore learn task-dependent reuse quality from selected feedback while controlling quality loss and execution cost.

\emph{\uline{Challenge II: Coupled scheduling}} arises because tasks share the cost of preparing reusable results, coupling their reuse decisions. The scheduler must therefore jointly choose which results to prepare and how to assign tasks to them, accounting for shared preparation costs, task-dependent execution costs, and quality losses.

We formulate cross-model computation reuse as an online decision problem and develop \emph{Quality-Aware Reuse Scheduling} (\alg{}) to address it. For Challenge I, \alg{} learns task-dependent reuse quality from selected, possibly delayed feedback and uses optimistic estimates to guide reuse decisions. For Challenge II, its Joint Adaptive Scheduling (\textsc{JAS}) algorithm jointly chooses which results to prepare and how to assign tasks to them. The scheduler uses a computable optimization gap to adjust scheduling accuracy to the remaining quality uncertainty.

Our main contributions are as follows:
\begin{itemize}[leftmargin=*]
\item \textbf{Problem formulation.} We provide a theoretical framework for cross-model computation reuse. We formulate it as an online decision problem, capturing quality uncertainty and coupled scheduling. We show that the shared preparation structure makes scheduling NP-hard even with known quality, breaking the efficient-oracle assumption in existing methods (Section~\ref{sec:model}).

\item \textbf{Joint scheduling with adaptive accuracy.} We develop \alg{}, which learns reuse quality from selected, possibly delayed feedback and uses \textsc{JAS} for joint preparation and assignment. Unlike existing methods, \textsc{JAS} bounds each feasible schedule's optimization error by a computable gap and stops once this gap matches the remaining quality uncertainty (Section~\ref{sec:online}).

\item \textbf{Theoretical guarantees.} We separate learning and optimization error, obtaining $\widetilde O(\sqrt{T})$ regret when quality-aware stopping is met at every epoch. Under the component conditions, we bound the scheduling work needed to reach a positive stopping tolerance (Section~\ref{sec:theory}).

\item \textbf{Empirical validation.} Experiments on four datasets support the analysis, with \alg{} reducing the combined cost of computation and quality loss by up to 18.0\%, and mean regret by 63.9\% over the strongest scheduling baseline (Section~\ref{sec:experiments}).
\end{itemize}
\section{Related Work}\label{sec:related}
\subsection{Computation Reuse}\label{sec:related-reuse}
Existing work on computation reuse mainly focuses on practical execution. Within a single model, Prompt Cache and SGLang reuse attention states across prompts \citep{gim2024promptcache,zheng2024sglang}, while CacheBlend and KVCOMM extend reuse to document segments and overlapping agent contexts \citep{yao2025cacheblend,ye2025kvcomm}. Across models, prompt compression and summarization let a small model shorten the input of a larger one \citep{jiang2023llmlingua,xu2024recomp}. At the cache level, DroidSpeak selectively recomputes parts of a cache shared across fine-tuned variants, and Cache-to-Cache learns to transform and combine source and target caches \citep{liu2026droidspeak,fu2026cachetocache}. These studies show how to produce and use reusable results, but not which results to prepare or which tasks should use them. We address this gap by studying joint result preparation and task assignment under unknown task-dependent quality and shared preparation costs, with guarantees on effectiveness and efficiency.

\subsection{Online Learning}\label{sec:related-learning}
Online learning addresses sequential decision-making under uncertainty by learning from observed feedback. Linear bandit methods use optimistic confidence bounds to guide exploration \citep{abbasi2011linear}. Contextual linear semi-bandits extend this principle to feasible combinations \citep{wen2015comblin,takemura2021contextual}, while delayed-feedback analyses account for pending observations \citep{howson2023qars,zierahn2025delayed}. AESCB lets the index approximation error scale with the confidence widths \citep{cuvelier2021aescb}. Other methods exploit matroid structure \citep{tzeng2024matroid} or reduce oracle calls \citep{kim2025oracle}. These analyses rely on a subroutine with an a priori guarantee over the action set. In our problem, shared preparation costs make even the optimistic problem NP-hard, so a polynomial-time exact oracle cannot be assumed unless $\mathrm{P}=\mathrm{NP}$. We instead certify each schedule by a computable additive gap to a global lower bound, stop once this gap matches the selected assignments' widths, and bound the work needed to reach it under the component conditions. We provide more detailed mathematical distinctions in Appendix \ref{app:directional-proof}.

\subsection{Coupled Scheduling}\label{sec:related-scheduling}
Coupled scheduling addresses interdependent decisions when tasks share preparation costs or compete for resources. Facility-location models capture shared opening costs and task-level assignment costs \citep{pasteris2021facility}. OSTOR uses primal--dual decomposition and adaptive dual descent to schedule reusable queries with approximation guarantees \citep{cheng2025ostor}. In LLM serving, LMetric combines cache reuse with load balancing \citep{zhang2026lmetric}, RouteLLM learns quality--cost tradeoffs for model selection \citep{ong2024routellm}, and Unified AI Gateway jointly formulates model, cache and placement decisions under quality and resource constraints \citep{lu2026gateway}. Our focus is on joint result preparation and task assignment with shared preparation costs and unknown task-dependent quality. The key distinction is that quality is learned from selected, possibly delayed feedback, and computable global gaps adapt scheduling accuracy to the remaining quality uncertainty while preserving resource feasibility.


\section{System Model and Problem Formulation}\label{sec:model}
In this section, we present the system model and formulate the optimization problem for cross-model computation reuse. Table~\ref{tab:notation} in Appendix~\ref{app:notation} summarizes the key notation.

\subsection{System Model}\label{sec:system-model}
\begin{wrapfigure}{r}{0.50\textwidth}
\centering
\includegraphics[width=\linewidth]{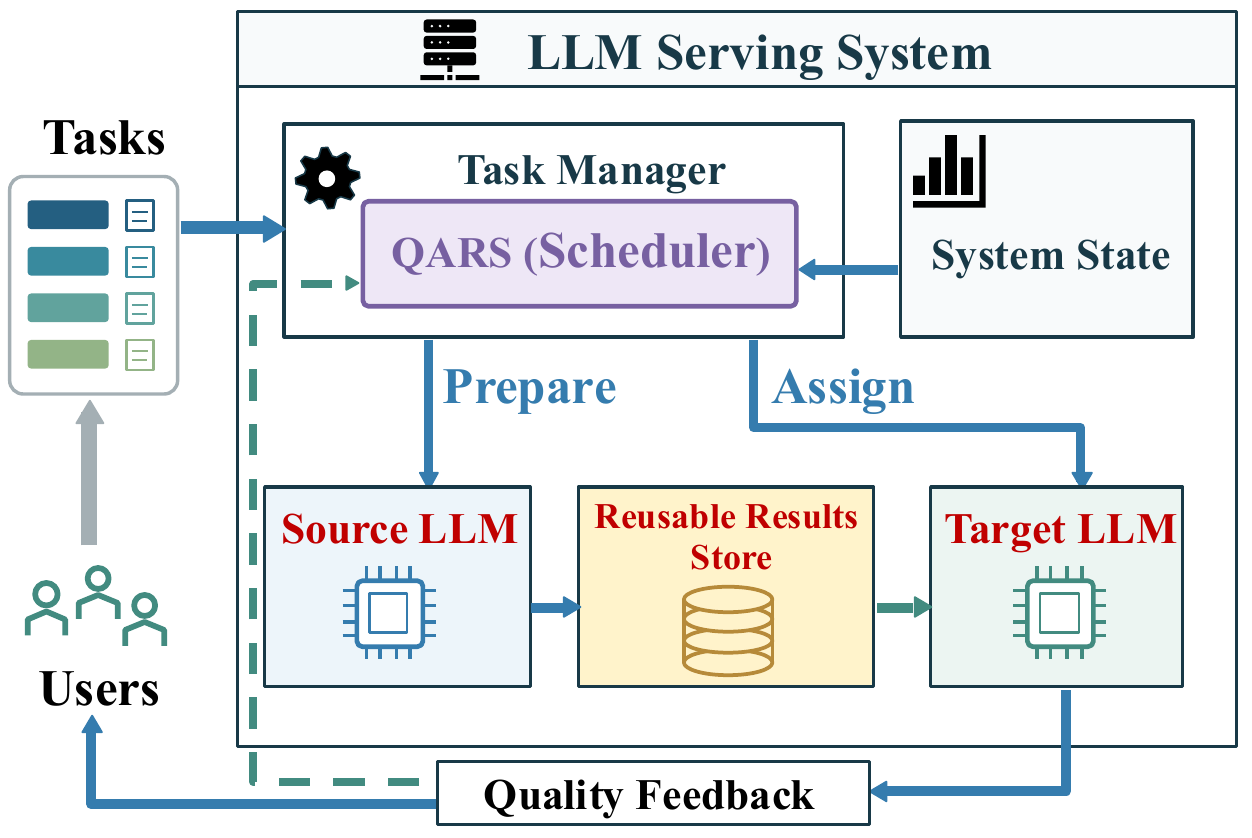}
\caption{System Overview.}
\label{fig:system-architecture}
\vspace{-0.4cm}
\end{wrapfigure}

Figure~\ref{fig:system-architecture} illustrates an LLM serving system supported by model-serving engines such as vLLM and SGLang \citep{kwon2023vllm,zheng2024sglang}, with \alg{} as its reuse scheduler. The system includes a task manager, source and target LLMs, and a reusable results store. \emph{Source} and \emph{target} denote the roles of computing and using a reusable result. The same model may play both roles. Cross-model reuse requires a compatible representation or a supported conversion \citep{liu2026droidspeak,fu2026cachetocache}.

Within the task manager, \alg{} uses the system state and quality feedback to jointly make scheduling (i.e., preparation and assignment) decisions. Based on these decisions, source models prepare selected results as needed and place them in the reusable results store. Target models execute tasks using the assigned results, and their outputs are evaluated to provide quality feedback, possibly after a delay. We consider \emph{decision epochs} $t\in\{1,\ldots,T\}$, each corresponding to one scheduling step, as modeled below:

\paragraph{Scheduling inputs.}
At the start of epoch $t$, the \emph{task set} $\cR_t$ contains the tasks to be scheduled, with $|\cR_t|\leq M$ for a \emph{workload bound} $M$. Each \emph{task} $r\in\cR_t$ is an admitted unit of work with a designated \emph{target model} $m_r^{\mathrm{tgt}}$. A \emph{reusable result} $g$ is a complete intermediate output for a specified input, computed by its \emph{source model} $m_g^{\mathrm{src}}$ under a preparation recipe that specifies its representation and any conversion needed for reuse. Examples include a document summary and a cached KV prefix \citep{liu2026droidspeak}. The \emph{result catalogue} $\cG_t$ is the finite set of results retained or available for preparation at epoch $t$.

A \emph{compatible pair} $(r,g)$ means that task $r$ can execute using result $g$, and the \emph{compatibility set} $\cE_t\subseteq\cR_t\times\cG_t$ collects these pairs.\footnote{Reuse is \emph{cross-model} when $m_g^{\mathrm{src}}\ne m_r^{\mathrm{tgt}}$. This setting also supports native execution (i.e., execution without computation reuse) and same-model reuse.} A \emph{context vector} $\boldsymbol{\phi}_{rg,t}\in\R^{D_\phi}$ describes each compatible pair $(r,g)$. It encodes the source--target model identities (i.e., $m_g^{\mathrm{src}}$ and $m_r^{\mathrm{tgt}}$), preparation recipe (e.g., the requested summary length) and task context (e.g., document length and question--document overlap). The model pair and preparation recipe together define a \emph{route}. Constructed from model settings and task inputs, the vector is known before scheduling and satisfies $\|\boldsymbol{\phi}_{rg,t}\|_2\leq1$.

There are two types of computation costs: \emph{preparation costs} and \emph{execution costs}. The \emph{preparation cost} $F_{g,t}\geq0$ is the work required to make $g$ available. It is paid once per selected result at epoch $t$, regardless of how many tasks use it, and is zero for already available results.  The \emph{execution cost} $e_{rg,t}\geq0$ measures the work to execute task $r$ using $g$, such as reading a summary, or loading cached states and recomputing selected layers. Both costs use the same work units and can be estimated through profiling \citep{zhong2024distserve}. They are treated as known before scheduling. Native execution incurs its full per-task execution cost with no shared preparation charge.

For $P$ \emph{resource types}, such as GPU compute capacity and memory, the \emph{result reservation} $\boldsymbol{c}_{g,t}\in\R^P$ and the \emph{task reservation} $\boldsymbol{c}_{rg,t}\in\R^P$ specify the nonnegative resources reserved for making result $g$ available and executing task $r$ using $g$, respectively, with $\bar{\boldsymbol{c}}_t\in\R^P$ denoting the available \emph{capacity}.

The \emph{observed state} $\mathsf{s}_t$ collects the scheduling inputs known before the decision: the task set $\cR_t$, result catalogue $\cG_t$, compatibility set $\cE_t$, context vectors $\boldsymbol{\phi}_{rg,t}$, costs $(F_{g,t},e_{rg,t})$, resource data $(P,\boldsymbol{c}_{g,t},\boldsymbol{c}_{rg,t},\bar{\boldsymbol{c}}_t)$, and a feasible fallback schedule.

\paragraph{Scheduling decisions and cost.}
Given the observed state $\mathsf{s}_t$ and the feedback already revealed, the scheduler jointly chooses preparation and assignment decisions $\boldsymbol{a}=(\boldsymbol{x},\boldsymbol{y})$ at epoch $t$. The \emph{preparation decision} $y_g=1$ selects result $g\in\cG_t$ for use, preparing it if it is not already available. The \emph{task assignment} $x_{rg}=1$ assigns task $r$ to result $g$ for a compatible pair $(r,g)\in\cE_t$. Let $\ell_{rg,t}$ denote the unknown \emph{expected quality loss} of this assignment conditional on $\mathsf{s}_t$. For a \emph{quality weight} $\omega>0$, the expected quality-aware cost is
\begin{equation}
 J_t(\boldsymbol{a})=
 \underbrace{\sum_{g\in\cG_t}F_{g,t}y_g}_{\text{preparation cost}}
 +\underbrace{\sum_{(r,g)\in\cE_t}e_{rg,t}x_{rg}}_{\text{execution cost}}
 +\underbrace{\omega\sum_{(r,g)\in\cE_t}\ell_{rg,t}x_{rg}}_{\text{weighted quality loss}}.
 \label{eq:objective}
\end{equation}
The quality weight $\omega$ converts expected quality loss into work units and balances computation cost and quality loss. Shared preparation costs are charged through $\boldsymbol{y}$, while task-dependent execution costs and quality losses enter through $\boldsymbol{x}$. Appendix~\ref{app:omega-sweep} provides evaluation under different $\omega$.

\paragraph{Quality feedback after scheduling.}
After execution, the \emph{observed quality loss} $\ell^{\mathrm{o}}_{rg,t}\in[0,1]$ evaluates the output of task $r$ using result $g$, for example as one minus a normalized task score. Its mean conditional on $\mathsf{s}_t$ is $\ell_{rg,t}$. Native execution and same-model reuse may also incur quality loss. Following linear bandit and semi-bandit models \citep{abbasi2011linear,wen2015comblin}, we model\footnote{Generalized linear \citep{li2017glmbandit} and neural \citep{zhou2020neuralucb} bandits offer richer quality models. Extending the confidence and regret analyses to them is future work.} $\ell_{rg,t}=\boldsymbol{\phi}_{rg,t}^{\top}\boldsymbol{\theta}^\star$ with $\|\boldsymbol{\theta}^\star\|_2\leq S$, for unknown $\boldsymbol{\theta}^\star$ and a known \emph{parameter-norm bound} $S$. Feedback is available only for selected assignments and may be delayed. For a known \emph{feedback-delay bound} $D$, each assignment selected at epoch $t$ yields one observation revealed before epoch $t+D+1$. Thus, $D=0$ means that feedback arrives before the next epoch, and at most $MD$ observations are pending before each scheduling decision. We make the following noise assumption \citep{abbasi2011linear,howson2023qars}:

\begin{assumption}[Feedback noise]\label{ass:online}
Given the information available before each reveal, $\varepsilon=\ell^{\mathrm{o}}-\boldsymbol{\phi}^\top\boldsymbol{\theta}^\star$ is mean-zero and $\sigma$-sub-Gaussian for a known \emph{noise scale} $\sigma$.
\end{assumption}

Binary task errors and normalized scores in $[0,1]$ satisfy the noise assumption with $\sigma=1/2$ by Hoeffding's lemma, provided the residuals are conditionally mean-zero. For observations sharing a result, the block extension in Appendix~\ref{app:block-feedback} accounts for their dependence with noise scale $\sigma\sqrt{\chi}$, where $\chi$ is a \emph{dependence factor}.

\subsection{Problem Formulation}\label{sec:problem}
\paragraph{Per-epoch scheduling problem $\mathcal{P}_t$.}
At epoch $t$, we minimize the quality-aware cost \eqref{eq:objective}:
\begin{subequations}
\renewcommand{\theequation}{\theparentequation-\arabic{equation}}
\begin{align}
 \qquad
 \min_{\boldsymbol{a}=(\boldsymbol{x},\boldsymbol{y})}\quad
 &J_t(\boldsymbol{a})
 \label{eq:epoch-problem}\tag{$\mathcal{P}_t$}\\
 \text{s.t.}\quad
 &x_{rg}\leq y_g,\qquad (r,g)\in\cE_t,
 \label{eq:preparation}\\
 &\sum_{g\in\cG_t}\boldsymbol{c}_{g,t}y_g
   +\sum_{(r,g)\in\cE_t}\boldsymbol{c}_{rg,t}x_{rg}
   \leq \bar{\boldsymbol{c}}_t,
 \label{eq:capacity}\\
  &\sum_{g:(r,g)\in\cE_t}x_{rg}=1,\qquad r\in\cR_t,
 \label{eq:assignment}\\
 &x_{rg}\in\{0,1\}\ ((r,g)\in\cE_t),\quad
   y_g\in\{0,1\}\ (g\in\cG_t).
 \label{eq:binary}
\end{align}
\end{subequations}

Constraint~\eqref{eq:preparation} links each assignment to a prepared result, allowing tasks to share its preparation cost. Constraint~\eqref{eq:capacity} enforces resource capacities componentwise: the first sum accounts for reservations per selected result, and the second for reservations per assignment. Each task is modeled as a single target-model execution, with compatible results serving as alternative reuse options, such as different summaries of the same document. Constraint~\eqref{eq:assignment} therefore assigns every admitted task to exactly one compatible option, including native execution without reuse.\footnote{Tasks that jointly use multiple results require modeling their combined quality and extending the scheduling algorithm and analysis. We leave this extension to future work.} Both preparation and assignment decisions are binary by \eqref{eq:binary}.

\paragraph{Long-term scheduling problem $\mathbb{P}$.}
Across $T$ epochs, we minimize the expected regret:
\begin{equation}
 \min_{\pi\in\Pi}\quad
 \mathbb{E}_{\pi}[R_T]
 =
 \mathbb{E}_{\pi}\left[
 \sum_{t=1}^T
 \bigl[J_t(\boldsymbol{a}_t)-J_t(\boldsymbol{a}_t^\star)\bigr]
 \right].
 \label{eq:online-problem}\tag{$\mathbb{P}$}
\end{equation}
Here, $\Pi$ denotes the set of online policies that choose $\boldsymbol{a}_t\in\cA_t(\mathsf{s}_t)$ using only the current observed state and the feedback already revealed, where $\cA_t(\mathsf{s}_t)$ is the feasible set of $\mathcal{P}_t$. Expected quality losses remain unknown at decision time, and the expectation averages over randomness in the online process. The benchmark $\boldsymbol{a}_t^\star$ solves $\mathcal{P}_t$ with the true expected quality losses under the same $\mathsf{s}_t$. Thus, regret evaluates reuse decisions along the observed state sequence rather than comparing policies that induce different retention histories.

Problem $\mathcal{P}_t$ is computationally challenging because shared preparation costs couple task assignments. Problem $\mathbb{P}$ further requires learning unknown quality losses from delayed feedback to guide these decisions. We establish the following complexity result.

\begin{proposition}[Computational complexity]\label{prop:epoch-hard}
Problem $\mathcal{P}_t$ is NP-hard even with known quality losses.
\end{proposition}

Appendix~\ref{app:hardness-proof} gives the proof by a Set Cover reduction. Appendix~\ref{app:coalition-proof} shows that reassigning tasks together can make reuse worthwhile by sharing preparation costs, even when individual reassignments offer no benefit. This motivates joint scheduling to capture shared savings, with a computable global gap to bound the remaining optimization error. Next, Section~\ref{sec:online} presents the joint scheduling algorithm.

\section{Quality-Aware Reuse Scheduling}\label{sec:online}
\subsection{Algorithm Overview}\label{sec:algorithm-design}
\alg{} (Algorithm~\ref{alg:qars}) addresses the online problem $\mathbb{P}$ by combining quality learning with joint scheduling. For Challenge~I, it learns task-dependent reuse quality from selected, possibly delayed feedback and constructs confidence bounds $\ell_{rg,t}^{\pm}$ (Lines~\ref{line:qars-feedback}--\ref{line:qars-estimates}, Section~\ref{sec:learn}). For Challenge~II, Joint Adaptive Scheduling (\textsc{JAS}) uses optimistic quality estimates to jointly choose which results to prepare and how to assign tasks to them (Line~\ref{line:qars-schedule}, Section~\ref{sec:dual}). It returns a feasible schedule and a global lower bound. The gap between the schedule's optimistic cost and this bound determines when to stop refinement and bounds the optimization error in regret (Section~\ref{sec:theory}).\footnote{We also consider a Thompson-sampling variant of \alg{}, denoted \alg{}-TS. It replaces $\ell_{rg,t}^{-}$ with clipped predictions from a single Gaussian parameter draw centered at the ridge estimate, while retaining the same feedback updates, \textsc{JAS} and stopping threshold. The effectiveness of \alg{}-TS is also evaluated in Section~\ref{sec:experiments}.} We detail each part below.

\begin{algorithm}[t]
\caption{Quality-Aware Reuse Scheduling (\alg{})}\label{alg:qars}
\KwIn{Confidence parameters $(\lambda_0,S,\sigma,\delta)$,
quality weight $\omega>0$, accuracy factor $\alpha>0$ and optional deadline.}
\KwOut{Feasible decision $\boldsymbol{a}_t$ at each epoch.}
\textbf{Initialize} feedback history $\mathcal H_0\leftarrow\varnothing$\;
\For{$t=1,\ldots,T$}{
Receive the observed state $\mathsf{s}_t$: $\cR_t,\cG_t,\cE_t$, $F_{g,t},e_{rg,t}$,
$\boldsymbol{c}_{g,t},\boldsymbol{c}_{rg,t},\bar{\boldsymbol{c}}_t$, $\boldsymbol{\phi}_{rg,t}$\nllabel{line:qars-input}\;
Form $\mathcal H_t$ by appending newly revealed
$(\tau,r,g,\boldsymbol{\phi}_{rg,\tau},\ell^{\mathrm{o}}_{rg,\tau})$ to $\mathcal H_{t-1}$ in reveal order\nllabel{line:qars-feedback}\;
Use the context vectors $\boldsymbol{\phi}_{rg,t}$ to compute
$\ell_{rg,t}^{\pm},\beta_t$ by \eqref{eq:estimator}--\eqref{eq:confidence}\nllabel{line:qars-estimates}\;
$\boldsymbol{a}_t\leftarrow
\operatorname{JAS}(\mathsf{s}_t,\ell_{rg,t}^{\pm},\boldsymbol{V}_t,\beta_t;\alpha,\omega)$
(Algorithm~\ref{alg:joint-scheduling})\nllabel{line:qars-schedule}\;
Execute $\boldsymbol{a}_t$, request quality feedback and output $\boldsymbol{a}_t$\nllabel{line:qars-execute}\;
}
\end{algorithm}

\subsection{Quality Learning for Challenge I}\label{sec:learn}
\alg{} estimates task-dependent quality losses from revealed feedback and uses confidence bounds to quantify the remaining uncertainty. The \emph{feedback history} $\mathcal H_t$ contains each quality observation revealed before scheduling at epoch $t$ once, in reveal order. Let $\boldsymbol{\phi}_i,\ell^{\mathrm{o}}_i$ be the context vector and observed quality loss of record $i\in\mathcal H_t$. With \emph{regularization} $\lambda_0>0$, the \emph{observed-context matrix} $\boldsymbol{V}_t$ and the regularized estimator $\widehat{\boldsymbol{\theta}}_t$ are
\begin{equation}
 \boldsymbol{V}_t=\lambda_0 \boldsymbol{I}+\sum_{i\in\mathcal H_t}\boldsymbol{\phi}_i\boldsymbol{\phi}_i^\top,
 \qquad \widehat{\boldsymbol{\theta}}_t=\boldsymbol{V}_t^{-1}\sum_{i\in\mathcal H_t}\boldsymbol{\phi}_i\ell^{\mathrm{o}}_i.
 \label{eq:estimator}
\end{equation}
For positive definite $\boldsymbol{B}$, define $\|\boldsymbol{v}\|_{\boldsymbol{B}}=(\boldsymbol{v}^\top \boldsymbol{B}\boldsymbol{v})^{1/2}$. For a \emph{failure probability} $\delta\in(0,1)$, the \emph{confidence radius} $\beta_t$ and quality-loss bounds $\ell_{rg,t}^{\pm}$ are
\begin{align}
 \beta_t&=\sigma\sqrt{\log\frac{\det \boldsymbol{V}_t}{\det(\lambda_0 \boldsymbol{I})}
                    +2\log(1/\delta)}+\sqrt{\lambda_0}S,\nonumber\\
 \ell_{rg,t}^{\pm}&=\operatorname{clip}_{[0,1]}
       \left(\boldsymbol{\phi}_{rg,t}^{\top}\widehat{\boldsymbol{\theta}}_t
             \pm\beta_t\|\boldsymbol{\phi}_{rg,t}\|_{\boldsymbol{V}_t^{-1}}\right).
 \label{eq:confidence}
\end{align}
Clipping keeps the bounds within the quality-loss range $[0,1]$. \textsc{JAS} uses the lower loss bound $\ell_{rg,t}^-$ as an optimistic estimate to guide reuse decisions. Our contribution connects quality learning with joint scheduling: the uncertainty terms $\beta_t\|\boldsymbol{\phi}_{rg,t}\|_{\boldsymbol{V}_t^{-1}}$ for selected assignments set the stopping tolerance for the computable optimization gap. This connection adapts scheduling accuracy to the remaining quality uncertainty. Appendix~\ref{app:confidence-basic} establishes the confidence bounds, and Appendix~\ref{app:v7-confidence} gives the centered route-wise implementation used in the experiments.

\subsection{Joint Scheduling for Challenge II}\label{sec:dual}
Using the learned quality estimates, \textsc{JAS} jointly chooses results to prepare and assigns tasks to them (Algorithm~\ref{alg:joint-scheduling}). Replacing $\ell_{rg,t}$ with $\ell_{rg,t}^{\pm}$ in \eqref{eq:objective} defines the optimistic cost $J_t^-$ and pessimistic cost $J_t^+$. \textsc{JAS} improves the best feasible schedule, called the \emph{incumbent}, and tightens a global lower bound. The gap between its optimistic cost and this bound controls the optimization error.

\paragraph{Component structure.}
A \emph{component} groups tasks with reusable results from the same input, such as questions and alternative summaries of one document. Reuse stays within components, which share result-storage capacity, while task execution resources are allocated separately. Motivated by prompt-module and prefix reuse \citep{gim2024promptcache,zheng2024sglang} and shared KV-cache management \citep{kwon2023vllm}, we adopt the following assumption.

\begin{assumption}[Component structure]\label{ass:components}
The compatibility graph has $\nu$ disjoint components $(\cR_i,\cG_i)$ with $|\cG_i|\leq K_{\mathrm{c}}$. There is one scalar capacity, with nonnegative reservations charged only to used results and zero task reservations.
\end{assumption}

\paragraph{Rounded local scheduling.}
Under Assumption~\ref{ass:components}, \textsc{JAS} uses Rounded Local Scheduling (\textsc{RLS}, Algorithm~\ref{alg:rls}), reported as \alg{} (Bounded). Each covering result subset defines a state by assigning tasks to their cheapest compatible options under optimistic costs, preserving the integral optimum by Lemma~\ref{lem:subset-original}. Dynamic programming (DP) selects one state per component under the shared capacity, rounding costs above component minima and retaining the minimum exact reservation per rounded total. The feasible schedule and global lower bound differ by at most $\nu B/\kappa$ by Lemma~\ref{thm:component-solver}, where $B$ is the initial gap. Higher rounding precision $\kappa$ improves accuracy at greater computational cost.

\paragraph{General local scheduling.}
For broader constraints, \textsc{JAS} uses General Local Scheduling (\textsc{GLS}, Algorithm~\ref{alg:gls}), reported as \alg{} (General), which also handles component workloads. Accepted moves jointly change result selections and task assignments to share preparation costs, preserve feasibility and reduce optimistic cost, as detailed in Appendix~\ref{app:collective-moves}. Assignment and capacity prices give a Lagrangian lower bound and screen results that cannot improve the incumbent by Proposition~\ref{prop:dual-screening}. If needed, linear programming (LP) strengthens the bound and mixed-integer linear programming (MILP) finds feasible candidates. Both procedures retain feasible schedules and valid global lower bounds, but only \textsc{RLS} has the refinement computation bound. The state, confidence bounds, and $\boldsymbol{V}_t$ stay fixed within each call. Algorithm~\ref{alg:joint-scheduling} omits epoch subscripts. Appendices~\ref{app:branch-procedures}, \ref{app:general-solver} and~\ref{app:structured-proofs} give the procedures, \textsc{GLS} initialization and \textsc{RLS} proofs, respectively.

\begin{algorithm}[t]
\caption{Joint Adaptive Scheduling (\textsc{JAS})}
\label{alg:joint-scheduling}
\KwIn{$\mathsf{s}_t,\ell_{rg,t}^{\pm},\boldsymbol{V}_t,\beta_t$;
$\alpha,\omega>0$; optional deadline.}
\KwOut{Feasible decision $\boldsymbol{a}_t^0$,
with retained global lower bound $L_t^-$.}
Set $\widetilde e_{rg}\leftarrow e_{rg}+\omega\ell_{rg}^-$
and form $J^{\pm}$ by \eqref{eq:objective}\;
Initialize $\boldsymbol{a}^0\leftarrow$ feasible fallback\;
\eIf{the workload satisfies Assumption~\ref{ass:components}}{
$(\boldsymbol{a}^0,L^-)\leftarrow
\operatorname{RLS}(\mathsf{s}_t,\widetilde e_{rg},
\boldsymbol{a}^0,\boldsymbol{V}_t,\beta_t;\alpha,\omega)$
(Algorithm~\ref{alg:rls})\nllabel{line:jas-rls}\;
}{
$(\boldsymbol{a}^0,L^-)\leftarrow
\operatorname{GLS}(\mathsf{s}_t,\widetilde e_{rg},
\boldsymbol{a}^0,\boldsymbol{V}_t,\beta_t;\alpha,\omega)$
(Algorithm~\ref{alg:gls})\nllabel{line:jas-gls}\;
}
\Return{$\boldsymbol{a}_t^0\leftarrow\boldsymbol{a}^0$;
retain $L_t^-\leftarrow L^-$}\;
\end{algorithm}

\paragraph{Quality-aware stopping.}\label{sec:stopping}
Quality-aware stopping limits scheduling work by matching accuracy to quality uncertainty. Let $\boldsymbol{a}_t^0$ be the incumbent, $m_t^-=\min_{\boldsymbol{a}\in\cA_t(\mathsf{s}_t)}J_t^-(\boldsymbol{a})$ the optimistic optimum, and $L_t^-\leq m_t^-$ the global lower bound. The \emph{actual optimization error} satisfies
\begin{equation}
 0\leq\eta_t=J_t^-(\boldsymbol{a}_t^0)-m_t^-
 \leq J_t^-(\boldsymbol{a}_t^0)-L_t^-.
 \label{eq:gaps}
\end{equation}
The gap includes lower-bound slack $m_t^--L_t^-$, which may persist for an LP bound even at optimality. For an \emph{accuracy factor} $\alpha>0$, \textsc{JAS} stops when
\begin{equation}
 J_t^-(\boldsymbol{a}_t^0)-L_t^-\leq\alpha\omega\beta_t
 \sum_{(r,g):x_{rg,t}^0=1}
 \|\boldsymbol{\phi}_{rg,t}\|_{\boldsymbol{V}_t^{-1}}.
 \label{eq:dual-stop}
\end{equation}
The threshold expresses selected assignments' quality uncertainty in work units, keeping optimization error at the same scale.

After each completed update, \textsc{JAS} checks this condition and returns the retained feasible schedule and valid bound on success or at the deadline. Missing or unverified outputs leave both unchanged. Section~\ref{sec:regret} accounts for the remaining optimization error if the deadline is reached first.

\section{Theoretical Guarantees}\label{sec:theory}
In this section, we provide theoretical guarantees on decision quality and computational efficiency.

\subsection{Decision Quality}\label{sec:regret}
We analyze \alg{}'s decision quality for the per-epoch problem $\mathcal{P}_t$ and the long-term problem $\mathbb{P}$.

\begin{theorem}[Decision quality]\label{thm:directional-delay}
Under Assumption~\ref{ass:online} and $\lambda_0\geq\max\{1,M\}$, \alg{} satisfies, with probability at least $1-\delta$,
\begin{align}
 \mathcal{P}_t:\quad
 J_t(\boldsymbol{a}_t)-J_t(\boldsymbol{a}_t^\star)
 &\leq J_t^+(\boldsymbol{a}_t)-L_t^-,
 \quad t\leq T,
 \label{eq:excess-cost-bound}\\
 \mathbb{P}:\quad
 R_T&\leq\sum_{t\leq T}\eta_t
       +2\omega\overline\beta_T\Xi_T.
 \label{eq:directional-regret}
\end{align}
\end{theorem}

Appendix~\ref{app:directional-proof} provides the proof and defines the confidence-radius bound $\overline\beta_T$ and cumulative context uncertainty bound $\Xi_T$, including delayed feedback. The per-epoch bound is computable, while the regret bound separates learning and optimization error \citep{chen2013cmab,cuvelier2021aescb}. Completing \eqref{eq:dual-stop} gives $R_T\leq(2+\alpha)\omega\overline\beta_T\Xi_T$, keeping optimization error at the scale of learning uncertainty. If a deadline is reached first, the bound retains the residual optimization error $\eta_t$.

\begin{corollary}[Regret]\label{cor:horizon-regret}
Under the conditions of Theorem~\ref{thm:directional-delay}, suppose \eqref{eq:dual-stop} is satisfied at every nonempty epoch, with probability at least $1-\delta$,
\begin{equation}
 R_T=O\Bigl((2+\alpha)\,\omega
 \bigl(\sigma D_\phi\log T+S\sqrt{\lambda_0D_\phi\log T}\bigr)
 \sqrt{MT(1+MD/\lambda_0)}\Bigr)=\widetilde O(\sqrt T).
 \label{eq:horizon-rates}
\end{equation}
\end{corollary}
Appendix~\ref{app:horizon-regret-proof} provides the proof. The corollary shows that matching scheduling accuracy to quality uncertainty achieves sublinear regret under bounded feedback delay, without requiring exact optimization at every epoch.

\subsection{Computational Efficiency}\label{sec:compute-guarantee}
We quantify the tradeoff between scheduling accuracy and computational complexity for \textsc{RLS}. Complexity counts scheduling arithmetic operations.

\begin{theorem}[Scheduling accuracy and computational complexity]
\label{thm:compute-regret}
Under Assumption~\ref{ass:components}, at each nonempty epoch $t$, \textsc{RLS} at a completed integer resolution $\kappa_t\geq1$ returns a feasible schedule and a valid global lower bound satisfying
\begin{equation}
 0\leq\eta_t
 \leq J_t^-(\boldsymbol{a}_t)-L_t^-
 \leq\frac{\nu_t B_t}{\kappa_t}.
 \label{eq:component-main-gap}
\end{equation}
Its computational complexity is $O(G^{\mathrm{g}}_t+N^{\mathrm{r}}_t(\kappa_t+1))$. With $\tau_t>0$ and no deadline, \textsc{RLS} reaches \eqref{eq:dual-stop} through doubling refinement with computational complexity
\begin{equation}
 O\!\left(G^{\mathrm{g}}_t+N^{\mathrm{r}}_t\nu_t
 \left(1+\frac{B_t}{\tau_t}\right)\right).
 \label{eq:statistical-work}
\end{equation}
\end{theorem}

Appendix~\ref{app:compute-proof} provides the proof and parameter definitions. Here $\nu_t$, $N^{\mathrm{r}}_t$ and $G^{\mathrm{g}}_t$ denote the component count, retained-state count and state-generation work, respectively. The resolution $\kappa_t$ is the chosen or last completed resolution. A smaller initial gap $B_t$ or larger tolerance lower bound $\tau_t$ reduces the complexity bound. Centering makes this bound invariant to common cost offsets within each component. Appendix~\ref{app:compute-proof} compares the finite-instance bound \eqref{eq:statistical-work} with exact scheduling to identify when computation can be reduced.

\section{Experimental Evaluation}\label{sec:evidence}\label{sec:experiments}

\subsection{Experimental Settings}\label{sec:exp-settings}
\paragraph{Data and workload.}
We evaluate the proposed framework using workloads built from four datasets---QuALITY, QASPER, MultiRC and RACE \citep{pang2022quality,dasigi2021qasper,khashabi2018multirc,lai2017race}---with Qwen3.5-4B and Qwen3.5-9B serving as the source and target models, respectively.\footnote{This work primarily focuses on the theoretical analysis of quality-aware reuse scheduling. We leave evaluation across different model families and at larger model scales to future work.} Appendix~\ref{app:experimental-protocol} provides more details.

\paragraph{Baselines and metrics.}
We compare \alg{} with four baselines adapted to the reuse problem. UCB-Greedy assigns tasks sequentially using optimistic losses. UCB-LocalSearch improves this schedule through result additions, removals and swaps with joint task reassignment \citep{arya2004local,su2024facility}. Following recent bandit comparisons \citep{kim2024local,pires2025exploitation}, \alg{}-TS uses Gaussian parameter samples, while Point Estimate uses clipped ridge predictions updated online. Both retain \textsc{JAS}. Additional controls include Frozen Point Estimate (a fixed predictor), NU (native execution), and SMU (\alg{} restricted to native execution and target-model summaries). We also compare \textsc{RLS}, \textsc{GLS}, Exact DP and MILP for scheduling efficiency. We report task quality $1-\ell^{\mathrm{o}}$, modeled computation cost, quality-aware cost \eqref{eq:objective}, and regret \eqref{eq:online-problem}, evaluated semi-synthetically with losses sampled from the quality model (Section~\ref{sec:system-model}) fit to real LLM outputs.

\subsection{Decision Quality}\label{sec:exp-results}
\paragraph{Online updates and exploration.}\label{sec:exp-measured}
The results show that QARS achieves the lowest quality-aware cost (90.78), 18.0\% below Frozen Point, supporting online quality learning and exploration. Point Estimate and QARS-TS remain 10.2\% and 10.0\% higher, so both online fitting and optimistic exploration help. SMU without source-summary routes is 10.2\% higher, and native-only NU is 12.7\% higher. Shared preparation lowers its modeled computation cost from 50.1 to 20.5 s/task. Quality-weight experiments are in Appendix~\ref{app:omega-sweep}.

\begin{wraptable}{r}{0.4\textwidth}
\vspace{-20pt}
\centering
\scriptsize
\setlength{\tabcolsep}{2.4pt}
\caption{Average regret per task.}
\label{tab:v9-main-primary}
\begin{tabular}{@{}lccccc@{}}\toprule
 & \alg{} & \alg{}-TS & \shortstack{Point\\Est.} & \shortstack{UCB-\\LocalSearch} & \shortstack{UCB-\\Greedy}\\\midrule
Regret & \textbf{0.524} & 1.070 & 2.167 & 1.450 & 6.084\\
Time & \textbf{0.245} & 0.249 & 0.987 & 1.698 & 0.476\\
\bottomrule\end{tabular}
\end{wraptable}

\paragraph{Coupled scheduling with selected feedback.}
We compare five same-catalogue methods under selected feedback in Table~\ref{tab:v9-main-primary} and Figure~\ref{fig:online-v9-core}. \alg{} achieves the lowest mean regret with low scheduling time, demonstrating effective joint scheduling with selected feedback.

\begin{figure}[ht]\centering
\includegraphics[width=0.75\linewidth]{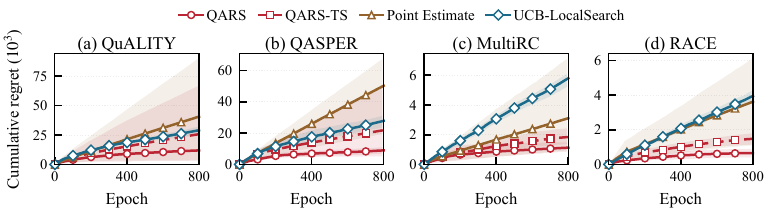}
\vspace{-6pt}
\caption{Cumulative regret across epochs.}
\label{fig:online-v9-core}
\end{figure}
\vspace{-5pt}

\subsection{Computational Efficiency}\label{sec:exp-time}
\paragraph{Quality-aware stopping.}
We evaluate quality-aware stopping in Table~\ref{tab:v9-ablation} and compare scheduling time across deadlines in Figure~\ref{fig:scheduling-time-deadline}. It reduces mean scheduling time by 58.8\% relative to zero tolerance, with comparable regret and quality-aware cost. \alg{}/\alg{}-TS have the shortest scheduling times across all tested deadlines. Appendices~\ref{app:solver-diagnostics} and~\ref{app:resource-sensitivity} provide detailed analyses.

\begin{figure}[ht]\centering
\includegraphics[width=0.75\linewidth]{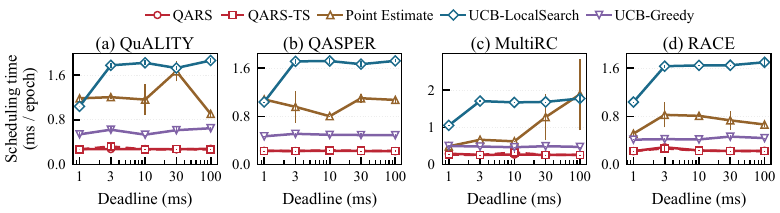}
\vspace{-6pt}
\caption{Scheduling time across deadlines for the five primary methods.}
\label{fig:scheduling-time-deadline}
\end{figure}
\vspace{-4pt}
\begin{wrapfigure}[7]{r}{0.45\textwidth}
\vspace{-10pt}
\centering
\includegraphics[width=\linewidth]{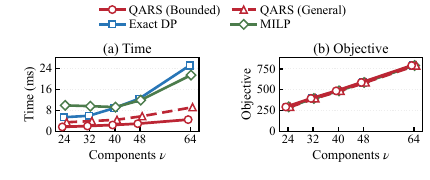}
\vspace{-25pt}
\caption{Time and objective versus $\nu$.}
\label{fig:online-v16-scale}
\label{fig:online-v16-scale-time}
\end{wrapfigure}

\paragraph{Joint adaptive scheduling.}
We compare Bounded (\textsc{RLS}) and General (\textsc{GLS}) with Exact DP and MILP on identical integer-component instances in Figure~\ref{fig:online-v16-scale}. Both \textsc{JAS} procedures reduce mean scheduling time ($2.72$ and $5.37$\,ms versus $11.55$ and $12.42$\,ms), while all four attain the same mean objective ($509.3$). Appendices~\ref{app:component-experiments} and~\ref{app:subset-assignment-exp} provide solver and catalogue comparisons.
\vspace{-4pt}
\section{Conclusion}\label{sec:conclusion}
We study online cross-model computation reuse with shared preparation costs, resource constraints and unknown task-dependent quality. \alg{} jointly chooses result preparations and task assignments, adapts scheduling accuracy to quality uncertainty, and attains $\widetilde O(\sqrt{T})$ 	regret with bounded scheduling work for component workloads. This work focuses on the theoretical framework. Future work will extend QARS to broader quality models, such as generalized linear and neural bandits, and to richer context features such as embeddings, and will evaluate it in larger-scale serving deployments, including KV-level cross-model reuse, more model pairs and end-to-end latency.

\label{zzz:endmain}\typeout{ENDMAINPOS \the\pagetotal\space OF \the\textheight}
\bibliography{references}
\bibliographystyle{conf2027_conference}
\clearpage
\appendix
\numberwithin{theorem}{section}
\numberwithin{proposition}{section}
\numberwithin{lemma}{section}
\numberwithin{corollary}{section}
\numberwithin{remark}{section}
\numberwithin{figure}{section}
\numberwithin{table}{section}

\noindent
\section{Notation}\label{app:notation}
Table~\ref{tab:notation} summarizes the notation.
\begingroup
\small
\setlength{\tabcolsep}{4pt}
\renewcommand{\arraystretch}{1.12}
\begin{longtable}{@{}>{\raggedright\arraybackslash}p{0.2\textwidth}p{0.8\textwidth}@{}}
\caption{Notation.}\label{tab:notation}\\
\toprule
Symbol & Meaning \\
\midrule
\endfirsthead
\multicolumn{2}{@{}l}{\textbf{Table~\thetable\ (continued)}}\\
\toprule
Symbol & Meaning \\
\midrule
\endhead
\midrule
\multicolumn{2}{r}{Continued on the next page}\\
\endfoot
\bottomrule
\endlastfoot

$t,\ T$ & Decision epoch and horizon.\\
$r,\ g,\ m_r^{\mathrm{tgt}},\ m_g^{\mathrm{src}}$ & Task, reusable result, target model, and source model.\\
$\cR_t,\ \cG_t,\ \cE_t$ & Task set, result catalogue, and compatibility set.\\
$\mathsf{s}_t,\ \cA_t(\mathsf{s}_t)$ & Observed state and its feasible set.\\
$\boldsymbol{a}=(\boldsymbol{x},\boldsymbol{y}),\ \boldsymbol{a}_t^0,\ \boldsymbol{a}_t^\star$ & Decision, incumbent, and per-epoch benchmark. $x_{rg}$ assigns task $r$ to result $g$, and $y_g$ selects result $g$.\\
\addlinespace[3pt]
$F_{g,t},\ e_{rg,t},\ \omega$ & Preparation cost, execution cost, and quality weight.\\
$J_t(\boldsymbol{a}),\ J_t^{\pm}(\boldsymbol{a})$ & Expected quality-aware cost, and the optimistic and pessimistic costs $J_t^-$ and $J_t^+$.\\
$R_T,\ \Pi$ & Cumulative regret along the observed states, and the set of online policies.\\
$P,\ \boldsymbol{c}_{g,t},\ \boldsymbol{c}_{rg,t},\ \bar{\boldsymbol{c}}_t$ & Number of resource types, result reservation, task reservation, and capacity. The scalar capacity is $c^{\max}$.\\
$\ell^{\mathrm{o}}_{rg,t},\ \ell_{rg,t},\ \ell_{rg,t}^{\pm}$ & Observed quality loss, expected quality loss, and confidence bounds.\\
$\boldsymbol{\phi}_{rg,t},\ \boldsymbol{\theta}^\star,\ D_\phi,\ S$ & Context vector, unknown parameter, context dimension, and parameter-norm bound, with $\ell_{rg,t}=\boldsymbol{\phi}_{rg,t}^{\top}\boldsymbol{\theta}^\star$.\\
$M,\ D,\ \chi$ & Workload bound, feedback-delay bound, and dependence factor.\\
$\mathcal{H}_t,\ \boldsymbol{V}_t,\ \widehat{\boldsymbol{\theta}}_t$ & Feedback history, observed-context matrix, and regularized estimator.\\
$\lambda_0,\ \sigma,\ \delta,\ \beta_t$ & Regularization, noise scale, failure probability, and confidence radius.\\
\addlinespace[3pt]
$\alpha,\ \tau_t$ & Accuracy factor and tolerance lower bound.\\
$\nu_t,\ (\cR_i,\cG_i),\ K_{\mathrm{c}},\ M_{\mathrm{c}}$ & Component count, component task and result sets, and maximum numbers of results and tasks per component. Within one epoch, $\nu_t$ is also written $\nu$.\\
$\mathcal{Q}_i,\ \mathcal{U},\ B_t,\ \kappa_t$ & Retained component states, a covering result subset, initial gap, and rounding resolution.\\
$N^{\mathrm{r}}_t,\ G^{\mathrm{g}}_t$ & Retained-state count and state-generation work.\\
\addlinespace[3pt]
$\Lambda_T,\ \overline\beta_T,\ H_T,\ \Xi_T$ & Logarithmic potential, confidence-radius bound, directional delay term, and cumulative context uncertainty bound.\\
\end{longtable}
\endgroup

\section{System Model and Problem Formulation}\label{app:model-details}
\subsection{Dependent Feedback Analysis}
\label{app:block-feedback}
Observations sharing a reusable result can be dependent. Blocks of jointly revealed feedback account for this dependence while preserving the linear quality model and regularized estimator.

\paragraph{Feedback blocks and noise conditions.}
Index the blocks in reveal order by $b=1,2,\ldots$. Block $b$ contains $n_b\leq N_{\max}^{\mathrm{blk}}$ observed quality losses from selected assignments, $\boldsymbol{\ell}^{\mathrm{o}}_b\in[0,1]^{n_b}$, with one context vector per row in $\boldsymbol{X}_b\in\R^{n_b\times {D_\phi}}$. Its membership, size and context matrix are measurable with respect to $\mathcal F_{b-1}$, the information available before the block is revealed. Each selected assignment contributes one observation to exactly one block, whose observations are revealed together. The feedback-delay and pending-observation bounds still apply. If batching increases either bound, the corresponding constants must be updated.

The same parameter determines the conditional means, $\mathbb E[\boldsymbol{\ell}^{\mathrm{o}}_b\mid\mathcal F_{b-1}]=\boldsymbol{X}_b\boldsymbol{\theta}^\star$, with $\boldsymbol{\epsilon}_b=\boldsymbol{\ell}^{\mathrm{o}}_b-\boldsymbol{X}_b\boldsymbol{\theta}^\star$. For a noise scale $\sigma>0$ and a dependence factor $\chi\geq1$, every $\boldsymbol{v}\in\R^{n_b}$ satisfies
\begin{equation}
 \mathbb E\!\left[\exp(\boldsymbol{v}^\top\boldsymbol{\epsilon}_b)\mid\mathcal F_{b-1}\right]
 \leq\exp\!\left(\frac{\sigma^2\chi}{2}\|\boldsymbol{v}\|_2^2\right).
 \label{eq:block-mgf}
\end{equation}
This conditional moment-generating-function (MGF) bound is stronger than a covariance bound. It allows dependence within blocks and a random reveal order, but need not hold when the next block is selected using its unrevealed outcomes.

Let $n_{\mathrm{blk}}(t)$ be the number of blocks revealed before epoch $t$. The regularized estimator is
\begin{equation}
 \boldsymbol{V}_t=\lambda_0\boldsymbol{I}+\sum_{b=1}^{n_{\mathrm{blk}}(t)}\boldsymbol{X}_b^\top \boldsymbol{X}_b,\qquad
 \widehat{\boldsymbol{\theta}}_t=\boldsymbol{V}_t^{-1}\sum_{b=1}^{n_{\mathrm{blk}}(t)}\boldsymbol{X}_b^\top \boldsymbol{\ell}^{\mathrm{o}}_b.
 \label{eq:block-estimator}
\end{equation}
Each revealed observation contributes once to the estimator, which excludes pending observations. With $\Lambda_T$ and $\Xi_T$ defined in Appendix~\ref{app:directional-proof}, the confidence radius and its uniform upper bound are
\begin{align}
 \beta_t^{(\chi)}
 &=\sigma\sqrt{\chi}\sqrt{\log\frac{\det \boldsymbol{V}_t}{\det(\lambda_0\boldsymbol{I})}
                                      +2\log(1/\delta)}
                                      +\sqrt{\lambda_0}S,\nonumber\\
 \overline\beta_T^{(\chi)}
 &=\sigma\sqrt{\chi}\sqrt{\Lambda_T+2\log(1/\delta)}
                                      +\sqrt{\lambda_0}S.
 \label{eq:block-radius}
\end{align}

Predictability, conditional means and the MGF bound replace Assumption~\ref{ass:online}. The linear quality model in Section~\ref{sec:system-model} is unchanged. Confidence requires $\lambda_0>0$, and Theorem~\ref{thm:directional-delay} requires $\lambda_0\geq\max\{1,M\}$.

\begin{proposition}[Block feedback]
\label{prop:block-feedback}
Under these conditions, the estimator~\eqref{eq:block-estimator} satisfies, with probability at least $1-\delta$, simultaneously for every epoch $t$,
\begin{equation}
 \|\widehat{\boldsymbol{\theta}}_t-\boldsymbol{\theta}^\star\|_{\boldsymbol{V}_t}\leq\beta_t^{(\chi)}.
 \label{eq:block-confidence}
\end{equation}
\end{proposition}
The dependence factor scales only the noise term in the confidence radius, leaving the estimator unchanged. On the same event, the confidence bounds yield the regret bound
\begin{equation}
 R_T\leq\sum_{t=1}^T\eta_t
       +2\omega\overline\beta_T^{(\chi)}\Xi_T.
 \label{eq:block-regret}
\end{equation}
If every nonempty epoch satisfies \eqref{eq:dual-stop} with $\beta_t^{(\chi)}$ in place of $\beta_t$,
\begin{equation}
 R_T\leq(2+\alpha)\omega\overline\beta_T^{(\chi)}
                           \Xi_T.
 \label{eq:block-stopping-regret}
\end{equation}

\begin{proof}
The effective noise scale is $\sigma_0=\sigma\sqrt{\chi}>0$. Let $\boldsymbol{A}_n=\sum_{b\leq n}\boldsymbol{X}_b^\top \boldsymbol{X}_b$ and $\boldsymbol{m}_n=\sum_{b\leq n}\boldsymbol{X}_b^\top\boldsymbol{\epsilon}_b$. Predictability of $\boldsymbol{X}_b$ allows $\boldsymbol{v}=\boldsymbol{X}_b\boldsymbol{u}$ in \eqref{eq:block-mgf}, so
\[
 z_n^{\mathrm{exp}}(\boldsymbol{u})=\exp\!\left(\boldsymbol{u}^\top \boldsymbol{m}_n-\frac{\sigma_0^2}{2}\boldsymbol{u}^\top \boldsymbol{A}_n\boldsymbol{u}\right)
\]
is a nonnegative supermartingale starting at one. Integrate this process against the Gaussian density $\boldsymbol{u}\sim N(\boldsymbol{0},(\sigma_0^2\lambda_0)^{-1}\boldsymbol{I})$, fixed before any feedback is revealed. Conditional Tonelli's theorem preserves the supermartingale property, and Gaussian integration gives the mixture
\[
 z_n^{\mathrm{mix}}=\left(\frac{\det(\lambda_0\boldsymbol{I})}{\det(\lambda_0\boldsymbol{I}+\boldsymbol{A}_n)}\right)^{1/2}
       \exp\!\left(\frac{\|\boldsymbol{m}_n\|_{(\lambda_0\boldsymbol{I}+\boldsymbol{A}_n)^{-1}}^2}
                              {2\sigma_0^2}\right).
\]
Ville's inequality then gives, with probability at least $1-\delta$, for every $n$,
\[
 \|\boldsymbol{m}_n\|_{(\lambda_0\boldsymbol{I}+\boldsymbol{A}_n)^{-1}}
 \leq\sigma_0\sqrt{\log\frac{\det(\lambda_0\boldsymbol{I}+\boldsymbol{A}_n)}{\det(\lambda_0\boldsymbol{I})}
                                      +2\log(1/\delta)}.
\]
At epoch $t$, take $n=n_{\mathrm{blk}}(t)$. Then $\widehat{\boldsymbol{\theta}}_t-\boldsymbol{\theta}^\star =\boldsymbol{V}_t^{-1}(\boldsymbol{m}_{n_{\mathrm{blk}}(t)}-\lambda_0\boldsymbol{\theta}^\star)$. Since $\boldsymbol{V}_t\succeq\lambda_0\boldsymbol{I}$, the triangle inequality bounds the regularization contribution by $\sqrt{\lambda_0}S$, proving \eqref{eq:block-confidence}.

Clipping preserves coverage of the expected quality loss. Appendix~\ref{app:directional-proof} uses only the selected contexts and reveal indicators, so the new radius yields \eqref{eq:block-regret} and \eqref{eq:block-stopping-regret}.
\end{proof}

\paragraph{Choosing the dependence factor.}
Suppose each observation lies in $[0,1]$ and has its stated conditional mean. The conditional range of $\boldsymbol{v}^\top\boldsymbol{\ell}^{\mathrm{o}}_b$ then has length at most $\|\boldsymbol{v}\|_1$, and conditional Hoeffding's lemma with Cauchy--Schwarz gives
\begin{equation}
 \mathbb E\!\left[e^{\boldsymbol{v}^\top\boldsymbol{\epsilon}_b}\mid\mathcal F_{b-1}\right]
 \leq\exp\!\left(\frac{\|\boldsymbol{v}\|_1^2}{8}\right)
 \leq\exp\!\left(\frac{N_{\max}^{\mathrm{blk}}\|\boldsymbol{v}\|_2^2}{8}\right).
 \label{eq:block-bounded}
\end{equation}
Thus $\sigma=1/2$ and $\chi=N_{\max}^{\mathrm{blk}}$ satisfy \eqref{eq:block-mgf} under arbitrary dependence within a block. Conditional independence given $\mathcal F_{b-1}$ allows $\chi=1$. An intermediate factor requires a conditional MGF bound, which sample correlation alone cannot establish.

\subsection{Proof of Computational Complexity}
\label{app:hardness-proof}
\begin{proof}[Proof of Proposition~\ref{prop:epoch-hard}]
We reduce Set Cover with $N_{\mathrm{set}}$ sets to $\mathcal{P}_t$. Given a universe $\{e_1,\dots,e_{N_{\mathrm{el}}}\}$ and sets $\mathcal S_1,\dots,\mathcal S_{N_{\mathrm{set}}}$ whose union is the universe, construct one epoch with a task $r_j$ per element and a result $g_i$ per set. Set $F_{g_i}=1$, $e_{r_jg_i}=0$, $\ell_{r_jg_i}=0$ and $\boldsymbol{c}_{g_i}=\boldsymbol{c}_{r_jg_i}=\boldsymbol{0}$, and let $(r_j,g_i)\in\cE$ exactly when $e_j\in\mathcal S_i$. Choose $\bar{\boldsymbol{c}}$ large enough that \eqref{eq:capacity} never binds; activating all set results is a feasible fallback because the sets cover the universe. Add a native option $g_0$ compatible with every task, with $F_{g_0}=0$, $e_{r_jg_0}=N_{\mathrm{set}}+1$, zero quality loss and zero reservations. Any schedule using it costs more than a schedule induced by all sets, of cost at most $N_{\mathrm{set}}$, so no optimum uses native execution.

A feasible schedule without native assignments consists of the selected set $\mathcal U=\{i:y_{g_i}=1\}$ together with an assignment of each task to a selected compatible result. By \eqref{eq:assignment} and \eqref{eq:preparation} such an assignment exists if and only if $\mathcal U$ covers the universe, and by construction $J(\boldsymbol{a})=|\mathcal U|$. The optimum of $\mathcal{P}_t$ therefore equals the minimum cover size. Since Set Cover is NP-hard, so is $\mathcal{P}_t$.
\end{proof}

\subsection{Reassigning Tasks Together}
\label{app:coalition-proof}
The following instance illustrates the joint savings discussed in Section~\ref{sec:problem}. Expected quality losses are zero and capacity is unconstrained, so costs reflect only preparation and execution. There are $n\geq 2$ tasks. Each native option has preparation cost $0$ and execution cost $1$, and one reusable result has preparation cost $3/2$ and execution cost $0$. The all-native schedule costs $n$. Reassigning one task to the reusable result raises the cost by $1/2$, so this individual reassignment offers no benefit. Reassigning all $n$ tasks together costs $3/2<n$, making reuse worthwhile by sharing the preparation cost.

\section{Quality Learning}\label{app:onlineproof}

\subsection{Simultaneous Confidence Bounds}\label{app:confidence-basic}
The regularized estimator uses only revealed quality feedback from selected assignments. Index these observations in reveal order by $i=1,2,\ldots$, with context vector $\boldsymbol{\phi}_i$ and observed quality loss $\ell^{\mathrm{o}}_i$. Let $\mathcal F_{i-1}$ be the information available before $\ell^{\mathrm{o}}_i$ is revealed, including $\boldsymbol{\phi}_i$. By Assumption~\ref{ass:online}, the residual $\varepsilon_i=\ell^{\mathrm{o}}_i-\boldsymbol{\phi}_i^\top\boldsymbol{\theta}^\star$ satisfies
\[
 \mathbb E[\varepsilon_i\mid\mathcal F_{i-1}]=0,\qquad
 \mathbb E[e^{u\varepsilon_i}\mid\mathcal F_{i-1}]
 \leq e^{\sigma^2u^2/2}\quad (u\in\mathbb R).
\]
These conditions are imposed in reveal order. A feedback-delay bound alone does not ensure them when that order depends on unrevealed quality losses.

After $n$ observations have been revealed, the cumulative noise vector and observed-context matrix are
\[
 \boldsymbol{m}_n=\sum_{i=1}^n\boldsymbol{\phi}_i\varepsilon_i,\qquad
 \boldsymbol{V}^{(n)}=\lambda_0\boldsymbol{I}+\sum_{i=1}^n\boldsymbol{\phi}_i\boldsymbol{\phi}_i^\top.
\]
For any $\lambda_0>0$, the uniform self-normalized inequality of \citet{abbasi2011linear} gives, with probability at least $1-\delta$,
\begin{equation}
 \|\boldsymbol{m}_n\|_{(\boldsymbol{V}^{(n)})^{-1}}
 \leq\sigma\sqrt{\log\frac{\det \boldsymbol{V}^{(n)}}{\det(\lambda_0 \boldsymbol{I})}
                         +2\log(1/\delta)}
 \quad\text{for every }n\geq0.
 \label{eq:scalar-noise-bound}
\end{equation}
At epoch $t$, $n_{\mathrm{o}}(t)=|\mathcal H_t|$ and $\boldsymbol{V}_t=\boldsymbol{V}^{(n_{\mathrm{o}}(t))}$. Substituting into \eqref{eq:estimator} gives
\[
 \widehat{\boldsymbol{\theta}}_t-\boldsymbol{\theta}^\star
 =\boldsymbol{V}_t^{-1}\bigl(\boldsymbol{m}_{n_{\mathrm{o}}(t)}-\lambda_0\boldsymbol{\theta}^\star\bigr).
\]
Since $\boldsymbol{V}_t\succeq\lambda_0\boldsymbol{I}$, the triangle inequality yields
\begin{align*}
 \|\widehat{\boldsymbol{\theta}}_t-\boldsymbol{\theta}^\star\|_{\boldsymbol{V}_t}
 &\leq\|\boldsymbol{m}_{n_{\mathrm{o}}(t)}\|_{\boldsymbol{V}_t^{-1}}
       +\lambda_0\|\boldsymbol{\theta}^\star\|_{\boldsymbol{V}_t^{-1}}\\
 &\leq\sigma\sqrt{\log\frac{\det \boldsymbol{V}_t}{\det(\lambda_0 \boldsymbol{I})}
                         +2\log(1/\delta)}+\sqrt{\lambda_0}S
 =\beta_t.
\end{align*}
On this event, denoted $\mathcal E_{\mathrm{conf}}$, weighted Cauchy--Schwarz gives, for every compatible pair,
\[
 \bigl|\boldsymbol{\phi}_{rg,t}^\top
    (\widehat{\boldsymbol{\theta}}_t-\boldsymbol{\theta}^\star)\bigr|
 \leq\beta_t\|\boldsymbol{\phi}_{rg,t}\|_{\boldsymbol{V}_t^{-1}}.
\]
Since $\ell_{rg,t}\in[0,1]$, clipping to this range preserves coverage of the expected quality loss:
\begin{equation}
 \|\widehat{\boldsymbol{\theta}}_t-\boldsymbol{\theta}^\star\|_{\boldsymbol{V}_t}\leq\beta_t,
 \qquad \ell_{rg,t}^-\leq\ell_{rg,t}\leq\ell_{rg,t}^+,
 \qquad 0\leq\ell_{rg,t}-\ell_{rg,t}^-
       \leq2\beta_t\|\boldsymbol{\phi}_{rg,t}\|_{\boldsymbol{V}_t^{-1}}.
 \label{eq:confidenceproof}
\end{equation}
Before clipping, the lower bound lies at most twice the radius below the conditional mean, and clipping at zero can only reduce this difference. The event holds for every $n$ and hence every epoch, without a further union bound. Confidence requires only $\lambda_0>0$. Theorem~\ref{thm:directional-delay} uses $\lambda_0\geq\max\{1,M\}$ to bound cumulative uncertainty over selected assignments.

\subsection{Centered Confidence Bounds for Routes}
\label{app:v7-confidence}
Let $N_{\mathrm{rt}}$ be the number of routes, with ${D_\phi}=2N_{\mathrm{rt}}$ and two coordinates per route. A compatible pair's context has two route coordinates $\boldsymbol{\psi}=(1-\xi,\xi)$ for $\xi\in[0,1]$, with all other coordinates zero. The primary experiments use four routes, so $N_{\mathrm{rt}}=4$. The additional endpoint assumption $\boldsymbol{\theta}_k^\star\in[0,1]^2$, with known center $\boldsymbol{\theta}_k^0=(1/2,1/2)$, gives $\|\boldsymbol{\theta}_k^\star-\boldsymbol{\theta}_k^0\|_2\leq1/\sqrt2$. Let $\mathcal H_{k,t}$ contain the records in $\mathcal H_t$ on route $k$. For any $\lambda_0>0$, the observed-context matrix, centered estimator, and confidence radius on this route are
\begin{align}
 \boldsymbol{V}_{k,t}&=\lambda_0\boldsymbol{I}_2+\sum_{i\in\mathcal H_{k,t}}\boldsymbol{\psi}_i\boldsymbol{\psi}_i^\top,\nonumber\\
 \widehat{\boldsymbol{\theta}}_{k,t}
 &=\boldsymbol{\theta}_k^0+\boldsymbol{V}_{k,t}^{-1}\sum_{i\in\mathcal H_{k,t}}\boldsymbol{\psi}_i(\ell^{\mathrm{o}}_i-1/2),\nonumber\\
 \beta_{k,t}&=\sigma\sqrt\chi
 \sqrt{\log\frac{\det \boldsymbol{V}_{k,t}}{\lambda_0^2}+2\log(N_{\mathrm{rt}}/\delta)}
 +\sqrt{\lambda_0/2}.
 \label{eq:v7-route-confidence}
\end{align}
Then $\boldsymbol{V}_t=\operatorname{diag}(\boldsymbol{V}_{1,t},\ldots,\boldsymbol{V}_{N_{\mathrm{rt}},t})$, and $\widehat{\boldsymbol{\theta}}_t$ concatenates the route estimators. This centered estimator replaces \eqref{eq:estimator}. Predictability, conditional linear means, and the vector-MGF condition~\eqref{eq:block-mgf} still apply. Routes index parameter coordinates, not blocks of jointly revealed feedback. The Gaussian-mixture argument of Appendix~\ref{app:block-feedback}, applied to each route with failure probability $\delta/N_{\mathrm{rt}}$, retains the noise scale $\sigma\sqrt\chi$ and bounds the regularization term by $\sqrt{\lambda_0/2}$. The resulting union bound does not require independence across routes. It gives an event $\mathcal E_{\mathrm{route}}$ of probability at least $1-\delta$ on which $\|\widehat{\boldsymbol{\theta}}_{k,t}-\boldsymbol{\theta}_k^\star\|_{\boldsymbol{V}_{k,t}}\leq\beta_{k,t}$ for every route and epoch.

For a compatible pair $(r,g)$ on route $k$, write $w_{rg,t}=\beta_{k,t}\|\boldsymbol{\phi}_{rg,t}\|_{\boldsymbol{V}_t^{-1}}$. Clip $\boldsymbol{\phi}_{rg,t}^\top\widehat{\boldsymbol{\theta}}_t\pm w_{rg,t}$ to $[0,1]$. For each coordinate $j$ on route $k$, also clip $\widehat\theta_{j,t}\pm\beta_{k,t}\sqrt{(\boldsymbol{V}_t^{-1})_{jj}}$ to $[0,1]$, collecting the bounds in $\underline{\boldsymbol{\theta}}_t$ and $\overline{\boldsymbol{\theta}}_t$. Intersect the prediction interval with $[\boldsymbol{\phi}_{rg,t}^\top\underline{\boldsymbol{\theta}}_t,\boldsymbol{\phi}_{rg,t}^\top\overline{\boldsymbol{\theta}}_t]$ to obtain $\ell_{rg,t}^-,\ell_{rg,t}^+$. On $\mathcal E_{\mathrm{route}}$, both intervals contain the conditional mean: nonnegative context coordinates preserve the endpoint bounds in the second interval. The intersection is nonempty, and $0\leq\ell_{rg,t}-\ell_{rg,t}^-\leq2w_{rg,t}$. Stopping uses the selected widths $w_{rg,t}$ in place of $\beta_t\|\boldsymbol{\phi}_{rg,t}\|_{\boldsymbol{V}_t^{-1}}$, preserving the excess-cost argument. The primary confirmation experiments use this route parameterization.

\section{Component Scheduling}\label{app:structured-proofs}
\subsection{Rounded Local Scheduling}\label{sec:component-algorithm}\label{app:component-exact}\label{app:rounding-proof}\label{app:component-regime-proof}
Under Assumption~\ref{ass:components}, omit epoch subscripts and write $c_g$ and ${c^{\max}}$ for the scalar result reservation and capacity. Each covering result subset $\mathcal U\subseteq\cG_i$ defines a state by assigning tasks to their cheapest compatible options under optimistic costs,
\begin{equation}
 J_i(\mathcal U)=\sum_{g\in \mathcal U}F_g+
     \sum_{r\in \cR_i}\min_{g\in \mathcal U:(r,g)\in\cE}\widetilde e_{rg},
 \qquad c_i(\mathcal U)=\sum_{g\in \mathcal U}c_g.
 \label{eq:component-state}
\end{equation}
Remove unused results and update each state's cost and reservation. Discard duplicate states and states with $c_i(\mathcal U)>{c^{\max}}$ to obtain the retained states $\mathcal Q_i$. These states preserve the integral optimum $m^-=\min_{\boldsymbol{a}\in\cA(\mathsf{s})}J^-(\boldsymbol{a})$ (Lemma~\ref{lem:subset-original}). The component minimum is $\underline J_i=\min_{\mathcal U\in\mathcal Q_i}J_i(\mathcal U)$. For an incumbent with optimistic cost $U_0$, the initial lower bound is $L_0=\sum_i\underline J_i$ and the initial gap is $B=U_0-L_0$. The retained-state count is $N^{\mathrm{r}}=\sum_i|\mathcal Q_i|$, and $G^{\mathrm{g}}$ bounds state-generation work. Dynamic programming selects one state per component under the shared capacity, rounding costs above component minima and retaining the minimum exact reservation per rounded total. The feasible schedule and global lower bound differ by at most $\nu B/\kappa$ (Lemma~\ref{thm:component-solver}).

\begin{algorithm}[H]
\caption{Rounded Local Scheduling (\textsc{RLS})}\label{alg:rls}
\KwIn{Observed state $\mathsf{s}$, optimistic assignment costs $\widetilde e_{rg}$,
feasible incumbent $\boldsymbol{a}^0$, $\boldsymbol{V},\beta$; $\alpha,\omega>0$; optional deadline.}
\KwOut{Feasible decision $\boldsymbol{a}^0$ and valid global lower bound $L^-$.}
Generate retained states $\mathcal Q_i$ once by \eqref{eq:component-state}\;
$\underline J_i\leftarrow\min_{\mathcal U\in\mathcal Q_i}J_i(\mathcal U)$\;
$U_0\leftarrow J^-(\boldsymbol{a}^0)$, $L_0\leftarrow\sum_i\underline J_i$,
$B\leftarrow U_0-L_0$, $L^-\leftarrow L_0$, $\kappa\leftarrow\nu$\;
\While{\eqref{eq:dual-stop} fails and time remains}{
$h\leftarrow B/\kappa$\;
$q_i(\mathcal U)\leftarrow\lfloor(J_i(\mathcal U)-\underline J_i)/h\rfloor$ for each $\mathcal U\in\mathcal Q_i$\;
Compute $f_i^{\mathrm{res}}(z)$ by \eqref{eq:rounded-component-dp}\;
$z_*\leftarrow\min\{z\leq \kappa:f_\nu^{\mathrm{res}}(z)\leq {c^{\max}}\}$\;
Reconstruct $\boldsymbol{a}_\kappa$ from $z_*$ and set $L_\kappa\leftarrow L_0+hz_*$\;
$\boldsymbol{a}^0\leftarrow\arg\min_{\boldsymbol{a}\in\{\boldsymbol{a}^0,\boldsymbol{a}_\kappa\}}J^-(\boldsymbol{a})$\;
$L^-\leftarrow\max\{L^-,L_\kappa\}$, $\kappa\leftarrow 2\kappa$\;
}
\Return{$\boldsymbol{a}^0,L^-$}\;
\end{algorithm}

For a component with ${K_{\mathrm{c}}}$ results and ${M_{\mathrm{c}}}$ tasks, there are at most $2^{K_{\mathrm{c}}}-1$ nonempty result subsets, compared with at most ${K_{\mathrm{c}}}^{M_{\mathrm{c}}}$ assignments.
\begin{lemma}[Exact subset representation]
\label{lem:subset-original}
Under Assumption~\ref{ass:components},
\[
 \min_{\{\mathcal U_i\}:\,\sum_i c_i(\mathcal U_i)\leq {c^{\max}}}\sum_iJ_i(\mathcal U_i)=m^-.
\]
\end{lemma}
\begin{proof}
For any feasible schedule, remove unused selected results without increasing cost or reservations, and let $\mathcal U_i$ contain the remaining results in component $i$. Assigning each task to its cheapest compatible result in $\mathcal U_i$ cannot increase its optimistic assignment cost, so the state optimum is at most $m^-$. Conversely, every feasible state selection gives compatible minimum-cost assignments. Removing unused results cannot increase preparation cost or reservations, so the resulting schedule is feasible and costs at most the state objective. The two optima are therefore equal. Discarding duplicate states and states with $c_i(\mathcal U)>{c^{\max}}$ preserves this value.
\end{proof}

\begin{lemma}[Component scheduling]
\label{thm:component-solver}
Under Assumption~\ref{ass:components}, every integer $\kappa\geq1$ yields a feasible $\widehat{\boldsymbol{a}}$ and global lower bound $L_\kappa$ satisfying
\begin{equation}
 L_\kappa\leq m^-\leq J^-(\widehat{\boldsymbol{a}}),\qquad
 J^-(\widehat{\boldsymbol{a}})-L_\kappa\leq\frac{\nu B}{\kappa},
 \label{eq:rounded-gap}
\end{equation}
in $O(G^{\mathrm{g}}+N^{\mathrm{r}}(\kappa+1))$ arithmetic operations.
\end{lemma}
\begin{proof}[Proof of Lemma~\ref{thm:component-solver}]
If $B=0$, the incumbent and $L_0$ coincide with $m^-$, so return the incumbent with $L_\kappa=L_0$. Otherwise set the rounding unit $h=B/\kappa$. The excess cost and rounded excess above the component minimum are
\begin{equation}
 \Delta_i(\mathcal U)=J_i(\mathcal U)-\underline J_i\geq0,\qquad
 q_i(\mathcal U)=\left\lfloor\frac{\Delta_i(\mathcal U)}h\right\rfloor.
\end{equation}
Discard states with $\Delta_i(\mathcal U)>B$. An optimal state selection has total excess at most $B$, so its rounded excess is at most $\kappa$.

The value $f_i^{\mathrm{res}}(z)$ is the minimum exact reservation of one state per component among the first $i$ components with total rounded excess $z$. Initialize $f_0^{\mathrm{res}}(0)=0$ and $f_0^{\mathrm{res}}(z)=+\infty$ otherwise, and compute
\begin{equation}
 f_i^{\mathrm{res}}(z)=\min_{\mathcal U\in\mathcal Q_i:\,\Delta_i(\mathcal U)\leq B,\ q_i(\mathcal U)\leq z}
       \{f_{i-1}^{\mathrm{res}}(z-q_i(\mathcal U))+c_i(\mathcal U)\},\qquad 0\leq z\leq \kappa.
 \label{eq:rounded-component-dp}
\end{equation}
Removing the last state leaves rounded total $z-q_i(\mathcal U)$; appending it adds reservation $c_i(\mathcal U)$. No resource coefficient is rounded. The smallest feasible rounded total is $z_* =\min\{z\leq \kappa:f_\nu^{\mathrm{res}}(z)\leq {c^{\max}}\}$, which exists because an optimal selection is retained. Reconstruct its states and remove unused results. Let $\widehat{\mathcal U}_i$ be this selection and $\boldsymbol{a}_\kappa$ the resulting feasible schedule, and set
\begin{equation}
 L_\kappa=L_0+h z_*.
 \label{eq:rounded-lower-bound}
\end{equation}
For an optimal selection $\mathcal U_i^*$,
\[
 L_\kappa\leq L_0+h\sum_i q_i(\mathcal U_i^*)
       \leq L_0+\sum_i\Delta_i(\mathcal U_i^*)=m^-.
\]
For the reconstructed selection $\widehat{\mathcal U}_i$, rounding down gives
\[
 J^-(\boldsymbol{a}_\kappa)\leq\sum_iJ_i(\widehat{\mathcal U}_i)
       < L_0+h(z_*+\nu)=L_\kappa+\nu B/\kappa.
\]
Let $\widehat{\boldsymbol{a}}$ be the cheaper of $\boldsymbol{a}_\kappa$ and the incumbent. Replacing $L_\kappa$ by any larger valid global lower bound can only reduce the gap, and \eqref{eq:rounded-gap} follows. Each state is processed for at most $\kappa+1$ rounded totals, so the arithmetic work is $O(G^{\mathrm{g}}+N^{\mathrm{r}}(\kappa+1))$.
\end{proof}

For fixed component size bounds $K_{\mathrm{c}},M_{\mathrm{c}}$ and resolution $\kappa$, this work is linear in the component count.
\begin{corollary}[Component complexity]
\label{cor:regime}
Under Assumption~\ref{ass:components}, $|\cG_i|\leq K_{\mathrm{c}}$ and $|\cR_i|\leq M_{\mathrm{c}}$, the component solver uses $O(2^{K_{\mathrm{c}}}\nu({K_{\mathrm{c}}}{M_{\mathrm{c}}}+\kappa))$ arithmetic operations.
\end{corollary}
\begin{proof}[Proof of Corollary~\ref{cor:regime}]
With $|\cG_i|\leq {K_{\mathrm{c}}}$ and $|\cR_i|\leq {M_{\mathrm{c}}}$, Lemma~\ref{thm:component-solver} gives $N^{\mathrm{r}}\leq\nu(2^{K_{\mathrm{c}}}-1)$ and $G^{\mathrm{g}}=\sum_i2^{|\cG_i|}|\cG_i|(|\cR_i|+1)\leq\nu2^{K_{\mathrm{c}}}{K_{\mathrm{c}}}({M_{\mathrm{c}}}+1)$, so $G^{\mathrm{g}}+N^{\mathrm{r}}(\kappa+1)=O(2^{K_{\mathrm{c}}}\nu({K_{\mathrm{c}}}{M_{\mathrm{c}}}+\kappa))$. Enumerating all preparation decisions instead considers $2^{\sum_i|\cG_i|}$ subsets.
\end{proof}

\subsection{Exact Dynamic Programming}\label{app:exact-capacity}\label{app:subset-proof}
For nonnegative integer reservations and integer capacity ${c^{\max}}$, Exact DP selects one retained state per component under the shared capacity. Let $f_i(c)$ be the minimum cost of the first $i$ components with total reservation exactly $c$. Initialize $f_0(0)=0$ and $f_0(c)=+\infty$ for $c\in\{1,\ldots,{c^{\max}}\}$, and compute
\begin{equation}
 f_i(c)=\min_{\mathcal U\in\mathcal Q_i:\,c_i(\mathcal U)\leq c}
        \{f_{i-1}(c-c_i(\mathcal U))+J_i(\mathcal U)\},\qquad 0\leq c\leq {c^{\max}}.
 \label{eq:exact-component-dp}
\end{equation}
Appending a state $\mathcal U$ to a selection of reservation $c-c_i(\mathcal U)$ gives total reservation $c$, so $m^-=\min_{c\leq {c^{\max}}}f_\nu(c)$. The recurrence takes $O(({c^{\max}}+1)N^{\mathrm{r}})$ arithmetic work, plus $O(G^{\mathrm{g}})$ to generate states. Because this work depends on the value of ${c^{\max}}$, Exact DP is pseudo-polynomial.

The component dispatch and scaling experiments use assignment catalogues rather than result-subset enumeration. Both preserve the integral optimum. With four tasks and four results, an assignment catalogue examines up to $256$ assignments and result-subset enumeration at most $15$ nonempty subsets. Exact DP combines retained states by capacity, and \textsc{RLS} by rounded excess cost. The bound $G^{\mathrm{g}}=O(\nu2^{K_{\mathrm{c}}}{K_{\mathrm{c}}}({M_{\mathrm{c}}}+1))$ applies to result-subset enumeration, not to these assignment catalogues. Reported times include catalogue construction and any overrun.

\section{Joint Scheduling}\label{app:joint-scheduling}
\subsection{General Local Scheduling}\label{app:branch-procedures}\label{sec:joint-bounds}\label{app:collective-moves}\label{app:dual-proof}
Fix an epoch and omit its subscript. Assignment and capacity prices give a Lagrangian lower bound. Let $[v]_+=\max\{v,0\}$, with assignment prices $\boldsymbol{u}\in\R^{|\cR|}$ and capacity prices $\boldsymbol{\pi}\geq0$. Define
\begin{align}
 b_g(\boldsymbol{u},\boldsymbol{\pi})&=F_g+\boldsymbol{\pi}^\top \boldsymbol{c}_g-
       \sum_{r:(r,g)\in\cE}[u_r-\widetilde e_{rg}-\boldsymbol{\pi}^\top \boldsymbol{c}_{rg}]_+,\nonumber\\
 L(\boldsymbol{u},\boldsymbol{\pi})&=\sum_r u_r-\boldsymbol{\pi}^\top \bar{\boldsymbol{c}}+\sum_g\min\{0,b_g(\boldsymbol{u},\boldsymbol{\pi})\}.
 \label{eq:dual-bound}
\end{align}
The term $b_g$ compares the preparation cost and priced resource reservation of result $g$ with the assignment benefit under these prices. Proposition~\ref{prop:dual-screening} gives $L(\boldsymbol{u},\boldsymbol{\pi})\leq m^-$ and screens results that cannot improve the incumbent.

\begin{algorithm}[H]
\caption{General Local Scheduling (\textsc{GLS})}\label{alg:gls}
\KwIn{Observed state $\mathsf{s}$, optimistic assignment costs $\widetilde e_{rg}$,
feasible incumbent $\boldsymbol{a}^0$, $\boldsymbol{V},\beta$; $\alpha,\omega>0$; optional deadline.}
\KwOut{Feasible decision $\boldsymbol{a}^0$ and valid global lower bound $L^-$.}
$u_r\leftarrow\min_{g:(r,g)\in\cE}\widetilde e_{rg}$,
$\boldsymbol{\pi}\leftarrow\boldsymbol{0}$\;
$L^-\leftarrow L(\boldsymbol{u},\boldsymbol{\pi})$ from \eqref{eq:dual-bound}\;
\While{\eqref{eq:dual-stop} fails and time remains}{
Update prices by \eqref{eq:dual-update}\;
$L^-\leftarrow\max\{L^-,L(\boldsymbol{u},\boldsymbol{\pi})\}$\;
Screen results by Proposition~\ref{prop:dual-screening}\;
Improve $\boldsymbol{a}^0$ by feasible collective moves with $\Delta J^-<0$
using \eqref{eq:collective-move}--\eqref{eq:collective-resource}\;
\If{\eqref{eq:dual-stop} still fails and time remains}{
Use LP/MILP to seek a feasible candidate $\widehat{\boldsymbol{a}}$ and a valid global bound $\widehat L$\;
If $\widehat{\boldsymbol{a}}$ is verified feasible, set
$\boldsymbol{a}^0\leftarrow\arg\min_{\boldsymbol{a}\in\{\boldsymbol{a}^0,\widehat{\boldsymbol{a}}\}}J^-(\boldsymbol{a})$\;
If $\widehat L$ is verified valid, set $L^-\leftarrow\max\{L^-,\widehat L\}$\;
}
}
\Return{$\boldsymbol{a}^0,L^-$}\;
\end{algorithm}

Accepted moves jointly change result selections and task assignments. Consider reassigning a nonempty compatible task group $\mathcal M$ to result $g$. Let $\iota(r)$ be the result currently assigned to task $r$, and let $\mathcal Z$ contain currently used results other than $g$ whose tasks all lie in $\mathcal M$. The move selects $g$ and removes the results in $\mathcal Z$, changing the optimistic cost by
\begin{equation}
 \Delta J^-=F_g(1-y_g)-\sum_{h\in\mathcal Z}F_h
       +\sum_{r\in\mathcal M}(\widetilde e_{rg}-\widetilde e_{r,\iota(r)}).
 \label{eq:collective-move}
\end{equation}
The first two terms account for changes in preparation costs, and the last for changes in execution cost and optimistic quality loss. The current resource reservation is $\boldsymbol{c}(\boldsymbol{a})=\sum_g\boldsymbol{c}_gy_g+\sum_{(r,g)\in\cE}\boldsymbol{c}_{rg}x_{rg}$, and the move changes it by
\begin{equation}
\Delta \boldsymbol{c}=\boldsymbol{c}_g(1-y_g)-\sum_{h\in\mathcal Z}\boldsymbol{c}_h
 +\sum_{r\in\mathcal M}(\boldsymbol{c}_{rg}-\boldsymbol{c}_{r,\iota(r)}).
\label{eq:collective-resource}
\end{equation}
A move is accepted only if the new assignments are compatible, $\boldsymbol{c}(\boldsymbol{a})+\Delta\boldsymbol{c}\leq\bar{\boldsymbol{c}}$, and $\Delta J^-<0$. It preserves feasibility and reduces the optimistic cost.

\begin{proposition}[Global lower bound and screening]\label{prop:dual-screening}
For any $\boldsymbol{u}\in\R^{|\cR|}$ and $\boldsymbol{\pi}\geq0$, $L(\boldsymbol{u},\boldsymbol{\pi})\leq m^-$. Every feasible schedule selecting $g$ has optimistic cost at least $L(\boldsymbol{u},\boldsymbol{\pi})+[b_g(\boldsymbol{u},\boldsymbol{\pi})]_+$.
\end{proposition}
Define the preparation and optimistic assignment costs adjusted by capacity prices as $F_g^{\boldsymbol{\pi}}=F_g+\boldsymbol{\pi}^\top \boldsymbol{c}_g$ and $\widetilde e_{rg}^{\boldsymbol{\pi}}=\widetilde e_{rg}+\boldsymbol{\pi}^\top \boldsymbol{c}_{rg}$.

\paragraph{Local minimization.}
Dualize the assignment equations by $\sum_r u_r(1-\sum_g x_{rg})$ and the capacity inequality by $\boldsymbol{\pi}^\top(\sum_g \boldsymbol{c}_g y_g+\sum_{r,g}\boldsymbol{c}_{rg}x_{rg}-\bar{\boldsymbol{c}})$. With the local constraints $x_{rg}\leq y_g$ and binary restrictions, the Lagrangian is
\begin{equation}
 \operatorname{lag}(\boldsymbol{x},\boldsymbol{y};\boldsymbol{u},\boldsymbol{\pi})=\sum_r u_r-\boldsymbol{\pi}^\top \bar{\boldsymbol{c}}+
 \sum_g\left(F_g^{\boldsymbol{\pi}} y_g+
             \sum_{r:(r,g)\in\cE}(\widetilde e_{rg}^{\boldsymbol{\pi}}-u_r)x_{rg}\right).
 \label{eq:dual-lagrangian}
\end{equation}
For each result $g$, $y_g=0$ contributes zero. If $y_g=1$, set $x_{rg}=1$ when $u_r>\widetilde e_{rg}^{\boldsymbol{\pi}}$ and $x_{rg}=0$ when $u_r<\widetilde e_{rg}^{\boldsymbol{\pi}}$, with either value at a tie, giving $b_g(\boldsymbol{u},\boldsymbol{\pi})$. The local minimum is $\min\{0,b_g(\boldsymbol{u},\boldsymbol{\pi})\}$, which yields \eqref{eq:dual-bound}.

\begin{proof}[Proof of Proposition~\ref{prop:dual-screening}]
For any feasible schedule $\boldsymbol{a}$, the assignment terms vanish and the capacity term in \eqref{eq:dual-lagrangian} is nonpositive, so
\[
 L(\boldsymbol{u},\boldsymbol{\pi})=\min_{\substack{x_{rg}\leq y_g\ ((r,g)\in\cE)\\\boldsymbol{x},\boldsymbol{y}\ \mathrm{binary}}}\operatorname{lag}(\boldsymbol{x},\boldsymbol{y};\boldsymbol{u},\boldsymbol{\pi})
 \leq\operatorname{lag}(\boldsymbol{a};\boldsymbol{u},\boldsymbol{\pi})\leq J^-(\boldsymbol{a}).
\]
A feasible primal LP objective alone need not lower-bound the integer minimum unless LP optimality is established. Forcing $y_g=1$ replaces $\min\{0,b_g\}$ by $b_g$, giving the lower bound $L+[b_g]_+$. If this exceeds $J^-(\bar{\boldsymbol{a}})$ for a feasible incumbent $\bar{\boldsymbol{a}}$, no optimal schedule under the optimistic costs can select $g$.
\end{proof}

\subsection{GLS Initialization and Price Updates}\label{app:oracle}\label{app:general-solver}
Before initializing the prices in Algorithm~\ref{alg:gls}, \textsc{GLS} compares the supplied feasible fallback, a cached feasible decision for the current template, and an incremental greedy decision under the current optimistic costs. It retains the cheapest completed feasible candidate. The initial lower bound uses $u_r=\min_{g:(r,g)\in\cE}\widetilde e_{rg}$ and $\boldsymbol{\pi}=\boldsymbol{0}$, giving $b_g=F_g\geq0$ and $L=\sum_ru_r$.

\paragraph{Price updates.}
Let $(\widetilde{\boldsymbol{x}},\widetilde{\boldsymbol{y}})$ be a minimizer of the relaxed Lagrangian for the current prices. A supergradient of $L$ is
\begin{equation}
 h_r=1-\sum_{g:(r,g)\in\cE}\widetilde x_{rg},\qquad
 \boldsymbol{s}^{\mathrm{cap}}=\sum_g \boldsymbol{c}_g\widetilde y_g+
           \sum_{(r,g)\in\cE}\boldsymbol{c}_{rg}\widetilde x_{rg}-\bar{\boldsymbol{c}}.
 \label{eq:dual-supergradient}
\end{equation}
At this minimizer,
\[
 L(\boldsymbol{u}',\boldsymbol{\pi}')
 \leq L(\boldsymbol{u},\boldsymbol{\pi})
       +\boldsymbol{h}^\top(\boldsymbol{u}'-\boldsymbol{u})
       +(\boldsymbol{s}^{\mathrm{cap}})^\top(\boldsymbol{\pi}'-\boldsymbol{\pi}).
\]
For any step size $\gamma_k>0$, the projected update is
\begin{equation}
 \boldsymbol{u}^{k+1}=\boldsymbol{u}^k+\gamma_k \boldsymbol{h}^k,\qquad
 \boldsymbol{\pi}^{k+1}=[\boldsymbol{\pi}^k+\gamma_k \boldsymbol{s}^{\mathrm{cap},\,k}]_+,
 \label{eq:dual-update}
\end{equation}
with the capacity prices projected coordinatewise. An update need not improve $L$, so the solver retains the largest valid lower bound seen so far.

\section{Decision Quality}\label{app:decision-quality}

\subsection{Proof of Theorem~\ref{thm:directional-delay}}
\label{app:directional-proof}

We use the linear quality model in Section~\ref{sec:system-model}. At each epoch, either \textsc{RLS} or \textsc{GLS} supplies \textsc{JAS} with a feasible schedule and a valid global lower bound. Let $N_T=\sum_{t\leq T}|\cR_t|$ be the total task count. Define the \emph{logarithmic potential} $\Lambda_T$ and the confidence-radius bound $\overline\beta_T$ as
\begin{equation}
 \Lambda_T={D_\phi}\log\!\left(1+\frac{N_T}{\lambda_0{D_\phi}}\right),\qquad
 \overline\beta_T=
 \sigma\sqrt{\Lambda_T+2\log(1/\delta)}+\sqrt{\lambda_0}S.
\end{equation}
The \emph{selected-context matrix} $\boldsymbol{\Phi}_t=\sum_{(r,g):x_{rg,t}=1}\boldsymbol{\phi}_{rg,t}\boldsymbol{\phi}_{rg,t}^{\top}$ collects the contexts selected at epoch $t$. The matrix $\boldsymbol{W}_t=\lambda_0\boldsymbol{I}+\sum_{s<t}\boldsymbol{\Phi}_s$ includes every selected assignment before epoch $t$, whereas $\boldsymbol{V}_t$ includes only revealed feedback. Their difference captures pending feedback, whose effect on selected assignments is measured by the \emph{directional delay term} $H_T$. Define
\begin{equation}
 H_T=\sum_{t\leq T}
 \tr\!\left((\boldsymbol{V}_t^{-1}-\boldsymbol{W}_t^{-1})\boldsymbol{\Phi}_t\right),\qquad
 \Xi_T=\sqrt{N_T(2\Lambda_T+H_T)}.
 \label{eq:directional-delay}
\end{equation}
Here $\Xi_T$ bounds cumulative context uncertainty. We write $\boldsymbol{a}_t=\boldsymbol{a}_t^0$ for the incumbent.

We first bound the per-epoch excess cost for $\mathcal{P}_t$, then regret for $\mathbb{P}$, accounting for selected contexts and pending feedback. All sums over $(r,g)$ at epoch $t$ refer to compatible pairs in $\cE_t$.

\begin{proof}
\par\smallskip\noindent\textbf{Step 1: Confidence containment.}\quad
Under Assumption~\ref{ass:online} in reveal order, the event $\mathcal E_{\mathrm{conf}}$ in \eqref{eq:confidenceproof} has probability at least $1-\delta$ for any $\lambda_0>0$, including the theorem's choice $\lambda_0\geq\max\{1,M\}$. On this event, Cauchy--Schwarz gives, for every compatible pair and epoch,
\[
 \left|\boldsymbol{\phi}_{rg,t}^{\top}
       (\widehat{\boldsymbol{\theta}}_t-\boldsymbol{\theta}^{\star})\right|
 \leq\beta_t\|\boldsymbol{\phi}_{rg,t}\|_{\boldsymbol{V}_t^{-1}}.
\]
Clipping preserves containment because expected quality losses lie in $[0,1]$:
\begin{equation}
 \ell_{rg,t}^{-}\leq\ell_{rg,t}\leq\ell_{rg,t}^{+},\qquad
 0\leq\ell_{rg,t}-\ell_{rg,t}^{-}
       \leq2\beta_t\|\boldsymbol{\phi}_{rg,t}\|_{\boldsymbol{V}_t^{-1}}.
 \label{eq:directional-confidence}
\end{equation}
This event covers contexts selected adaptively from the observed state and revealed feedback, without an additional union bound over decisions.

\par\smallskip\noindent\textbf{Step 2: Per-epoch decision quality.}\quad
Optimistic costs separate learning and optimization error. Since $J_t^-$, $J_t$ and $J_t^+$ differ only in quality loss, every feasible decision $\boldsymbol{a}$ satisfies $J_t^-(\boldsymbol{a})\leq J_t(\boldsymbol{a})\leq J_t^+(\boldsymbol{a})$ on $\mathcal E_{\mathrm{conf}}$. The global lower bound gives
\[
 L_t^-\leq m_t^-
       \leq J_t^-(\boldsymbol{a}_t^{\star})
       \leq J_t(\boldsymbol{a}_t^{\star}).
\]
Together with $J_t(\boldsymbol{a}_t)\leq J_t^+(\boldsymbol{a}_t)$, this proves the computable bound \eqref{eq:excess-cost-bound} for $\mathcal{P}_t$. Using the actual optimization error $\eta_t=J_t^-(\boldsymbol{a}_t)-m_t^-$ from \eqref{eq:gaps}, we obtain
\begin{align}
 J_t(\boldsymbol{a}_t)-J_t(\boldsymbol{a}_t^{\star})
 &=\bigl[J_t(\boldsymbol{a}_t)-J_t^-(\boldsymbol{a}_t)\bigr]
   +\eta_t+\bigl[m_t^--J_t(\boldsymbol{a}_t^{\star})\bigr]\nonumber\\
 &\leq\eta_t+\omega\sum_{(r,g):x_{rg,t}=1}
                      (\ell_{rg,t}-\ell_{rg,t}^-)\nonumber\\
 &\leq\eta_t+2\omega\beta_t
       \sum_{(r,g):x_{rg,t}=1}\|\boldsymbol{\phi}_{rg,t}\|_{\boldsymbol{V}_t^{-1}}.
 \label{eq:directional-epoch-regret}
\end{align}
Both decisions use the same feasible set $\cA_t(\mathsf{s}_t)$, as required by \eqref{eq:online-problem}.

\par\smallskip\noindent\textbf{Step 3: Selected-assignment potential.}\quad
The matrix $\boldsymbol{W}_t$ bounds uncertainty before accounting for delayed feedback. It satisfies $\boldsymbol{W}_{t+1}=\boldsymbol{W}_t+\boldsymbol{\Phi}_t$ and $\boldsymbol{W}_1=\lambda_0\boldsymbol{I}$. Each epoch selects one assignment for each admitted task, so the context-norm bound implies
\[
 0\preceq \boldsymbol{\Phi}_t\preceq|\cR_t|\boldsymbol{I}\preceq M\boldsymbol{I}\preceq\lambda_0\boldsymbol{I}\preceq \boldsymbol{W}_t.
\]
Thus the eigenvalues of $\boldsymbol{W}_t^{-1/2}\boldsymbol{\Phi}_t\boldsymbol{W}_t^{-1/2}$ lie in $[0,1]$. On this interval, $u\leq2\log(1+u)$ because $2\log(1+u)-u$ starts at zero and has derivative $(1-u)/(1+u)\geq0$. Applying this inequality gives
\begin{align*}
 \tr(\boldsymbol{W}_t^{-1}\boldsymbol{\Phi}_t)
 &=\tr(\boldsymbol{W}_t^{-1/2}\boldsymbol{\Phi}_t\boldsymbol{W}_t^{-1/2})\\
 &\leq2\log\det(\boldsymbol{I}+\boldsymbol{W}_t^{-1/2}\boldsymbol{\Phi}_t\boldsymbol{W}_t^{-1/2}) \\
  & =2\log\frac{\det \boldsymbol{W}_{t+1}}{\det \boldsymbol{W}_t}.
\end{align*}
Moreover, $\tr(\boldsymbol{W}_{T+1})\leq\lambda_0{D_\phi}+N_T$. The arithmetic--geometric mean inequality for its eigenvalues gives
\[
 \log\frac{\det \boldsymbol{W}_{T+1}}{\det(\lambda_0\boldsymbol{I})}
 \leq {D_\phi}\log\!\left(\frac{\tr(\boldsymbol{W}_{T+1})}{\lambda_0{D_\phi}}\right)
 \leq {D_\phi}\log\!\left(1+\frac{N_T}{\lambda_0{D_\phi}}\right)=\Lambda_T.
\]
Summing over epochs and telescoping gives
\begin{equation}
 \sum_{t=1}^T\tr(\boldsymbol{W}_t^{-1}\boldsymbol{\Phi}_t)
 \leq2\log\frac{\det \boldsymbol{W}_{T+1}}{\det(\lambda_0\boldsymbol{I})}
 \leq2\Lambda_T.
 \label{eq:delay-potential}
\end{equation}
This argument uses the full batch $\boldsymbol{\Phi}_t$, without requiring feedback between assignments in the same epoch.

\par\smallskip\noindent\textbf{Step 4: Pending feedback and selected uncertainty.}\quad
Pending feedback adds the nonnegative uncertainty term $H_T$. The matrix $\boldsymbol{P}_t=\boldsymbol{W}_t-\boldsymbol{V}_t$ sums the context outer products of observations pending before epoch $t$. Each revealed observation enters the history once, so $\boldsymbol{P}_t\succeq0$ and $\boldsymbol{V}_t\preceq\boldsymbol{W}_t$. Inversion gives $\boldsymbol{V}_t^{-1}-\boldsymbol{W}_t^{-1}\succeq0$, making every term in $H_T$ nonnegative. By the definition of $\boldsymbol{\Phi}_t$,
\begin{align}
 \sum_{t=1}^T\sum_{(r,g):x_{rg,t}=1}
       \|\boldsymbol{\phi}_{rg,t}\|_{\boldsymbol{V}_t^{-1}}^2
 &=\sum_{t=1}^T\tr(\boldsymbol{V}_t^{-1}\boldsymbol{\Phi}_t)\nonumber\\
 &=\sum_{t=1}^T\tr(\boldsymbol{W}_t^{-1}\boldsymbol{\Phi}_t)+H_T
 \leq2\Lambda_T+H_T.
 \label{eq:delay-potential-decomp}
\end{align}
Applying Cauchy--Schwarz to the $N_T$ selected assignments gives
\begin{equation}
 \sum_{t=1}^T\sum_{(r,g):x_{rg,t}=1}
       \|\boldsymbol{\phi}_{rg,t}\|_{\boldsymbol{V}_t^{-1}}
 \leq\sqrt{N_T(2\Lambda_T+H_T)}=\Xi_T.
 \label{eq:directional-width-sum}
\end{equation}
If $N_T=0$, both sides are zero. Finally, $\boldsymbol{V}_t\preceq \boldsymbol{W}_{T+1}$ and the same determinant bound used above imply
\[
 \beta_t
 \leq\sigma\sqrt{\Lambda_T+2\log(1/\delta)}+\sqrt{\lambda_0}S
 =\overline\beta_T.
\]
Summing \eqref{eq:directional-epoch-regret} and applying the radius bound and \eqref{eq:directional-width-sum} yields
\[
 R_T\leq\sum_{t=1}^T\eta_t
      +2\omega\overline\beta_T\Xi_T,
\]
which proves \eqref{eq:directional-regret} for $\mathbb{P}$ and separates optimization error from learning error. The matrix inequalities hold for every selected-context and reveal-time sequence allowed by the model. Both conclusions therefore hold on the same event $\mathcal E_{\mathrm{conf}}$, with probability at least $1-\delta$.
\end{proof}

\paragraph{Quality-aware stopping.}
Quality-aware stopping keeps optimization error at the scale of learning uncertainty. If \eqref{eq:dual-stop} is satisfied at every nonempty epoch, the valid lower bound $L_t^-$ gives
\[
 \eta_t\leq J_t^-(\boldsymbol{a}_t)-L_t^-
 \leq\alpha\omega\beta_t
       \sum_{(r,g):x_{rg,t}=1}\|\boldsymbol{\phi}_{rg,t}\|_{\boldsymbol{V}_t^{-1}}.
\]
Substituting this inequality into \eqref{eq:directional-epoch-regret} and summing with \eqref{eq:directional-width-sum} yields $R_T\leq(2+\alpha)\omega\overline\beta_T\Xi_T$. If a deadline is reached first, the theorem accounts for the remaining $\eta_t$, which can also be replaced by its computable upper bound $J_t^-(\boldsymbol{a}_t)-L_t^-$.

\paragraph{Pending feedback and selected contexts.}
Pending feedback changes the widths of the currently selected contexts. Define the normalized pending-context and selected-context matrices
\[
 \boldsymbol{Q}_t=\boldsymbol{V}_t^{-1/2}\boldsymbol{P}_t\boldsymbol{V}_t^{-1/2},\qquad
 \boldsymbol{Z}_t=\boldsymbol{V}_t^{-1/2}\boldsymbol{\Phi}_t\boldsymbol{V}_t^{-1/2}.
\]
Since $\boldsymbol{W}_t=\boldsymbol{V}_t^{1/2}(\boldsymbol{I}+\boldsymbol{Q}_t)\boldsymbol{V}_t^{1/2}$, matrix inversion and cyclicity of trace give
\begin{align*}
 \tr\!\left((\boldsymbol{V}_t^{-1}-\boldsymbol{W}_t^{-1})\boldsymbol{\Phi}_t\right)
 &=\tr\!\left([\boldsymbol{I}-(\boldsymbol{I}+\boldsymbol{Q}_t)^{-1}]\boldsymbol{Z}_t\right)\\
 &=\tr\!\left(\boldsymbol{Q}_t(\boldsymbol{I}+\boldsymbol{Q}_t)^{-1}\boldsymbol{Z}_t\right).
\end{align*}
The contribution is zero when $\boldsymbol{Q}_t\boldsymbol{Z}_t=0$. To bound it for every pattern of pending and selected contexts, let $\rho_t=\|\boldsymbol{Q}_t\|_{\mathrm{op}}$. The inequalities $\boldsymbol{P}_t\preceq\rho_t\boldsymbol{V}_t$ and $\boldsymbol{W}_t\preceq(1+\rho_t)\boldsymbol{V}_t$ imply
\[
 \boldsymbol{V}_t^{-1}-\boldsymbol{W}_t^{-1}\preceq\rho_t\boldsymbol{W}_t^{-1},\qquad
 \tr\!\left((\boldsymbol{V}_t^{-1}-\boldsymbol{W}_t^{-1})\boldsymbol{\Phi}_t\right)
 \leq\rho_t\tr(\boldsymbol{W}_t^{-1}\boldsymbol{\Phi}_t).
\]
At most $MD$ observations are pending, each with context norm at most one. Thus $\boldsymbol{P}_t\preceq MD\boldsymbol{I}$ and $\boldsymbol{V}_t\succeq\lambda_0\boldsymbol{I}$ imply $\rho_t\leq MD/\lambda_0$. Equation~\eqref{eq:delay-potential} then gives
\begin{equation}
 H_T\leq2\bigl(\max_{t\leq T}\rho_t\bigr)\Lambda_T
 \leq\frac{2MD}{\lambda_0}\Lambda_T,
 \qquad
 \Xi_T\leq
 \sqrt{2N_T(1+MD/\lambda_0)\Lambda_T}.
 \label{eq:directional-delay-envelope}
\end{equation}
Corollary~\ref{cor:horizon-regret} uses this bound. The realized $H_T$ depends on how selected contexts align with pending observations. A bound using only the pending count $MD$ omits this dependence.

\paragraph{Same pending count, different effect.}
The same number of pending observations can add different amounts of uncertainty, as a two-dimensional example shows. Let $\boldsymbol{V}_t=\operatorname{diag}(v_1,v_2)$ with $v_1,v_2\geq\lambda_0$, and suppose $n_{\mathrm{pend}}\leq MD$ pending observations all have context $\boldsymbol{e}_1$, so $\boldsymbol{P}_t=n_{\mathrm{pend}}\boldsymbol{e}_1\boldsymbol{e}_1^\top$. If $n_{\mathrm{cur}}\leq M$ current assignments all have context $\boldsymbol{e}_2$, their contribution to $H_T$ is zero. If their contexts instead equal $\boldsymbol{e}_1$, the contribution is
\[
 n_{\mathrm{cur}}\left(\frac{1}{v_1}-\frac{1}{v_1+n_{\mathrm{pend}}}\right)
 =\frac{n_{\mathrm{cur}}n_{\mathrm{pend}}}{v_1(v_1+n_{\mathrm{pend}})}.
\]
In this example, pending feedback adds uncertainty relative to $\boldsymbol{W}_t$ only when selected contexts align with the pending observations.

\paragraph{Dependent feedback blocks.}
Dependent feedback changes confidence widths while preserving the matrix argument, which uses only selected contexts and reveal times. Under the conditional block-feedback assumptions in Appendix~\ref{app:block-feedback}, Proposition~\ref{prop:block-feedback} provides the required simultaneous confidence event with radii $\beta_t^{(\chi)}$ and $\overline\beta_T^{(\chi)}$. Using these radii preserves the proof, the stopping guarantee and the delay bound. The dependence factor adjusts confidence widths, while $H_T$ captures pending feedback through the selected contexts.

\paragraph{Relation to approximate-oracle bandits.}\label{app:oracle-relation}
Approximate-oracle analyses bound regret by learning error plus oracle error. CUCB assumes an $(\alpha,\beta)$-approximation oracle and bounds regret against the $\alpha\beta$-scaled optimum \citep{chen2013cmab}. AESCB computes the ESCB index approximately, with error allowed to scale with the confidence widths, using budgeted linear maximization over the action set \citep{cuvelier2021aescb}. Theorem~\ref{thm:directional-delay} uses the same separation but differs in three respects. (i) \emph{Oracle.} Both analyses specify a priori approximation guarantees for their optimization subroutines. For $\mathcal{P}_t$, even linear minimization over the feasible set with known costs is Set-Cover-hard (Proposition~\ref{prop:epoch-hard}). Thus, a polynomial-time exact oracle cannot be assumed unless $\mathrm{P}=\mathrm{NP}$. We compare with the exact optimum $\boldsymbol{a}_t^\star$ and account for the additive optimization error. (ii) \emph{Certificate.} \textsc{JAS} returns a feasible schedule and a valid global lower bound, so $\eta_t\leq J_t^-(\boldsymbol{a}_t)-L_t^-$ holds at every epoch, including when a deadline interrupts the solver. (iii) \emph{Tolerance and work.} AESCB enforces its approximation tolerance by construction. \textsc{JAS} checks the optimization gap against the selected assignments' confidence widths in \eqref{eq:dual-stop}. For \textsc{RLS} under Assumption~\ref{ass:components}, Theorem~\ref{thm:compute-regret} bounds the work needed to pass this check when the tolerance lower bound is positive and no deadline is imposed.

\subsection{Proof of Corollary~\ref{cor:horizon-regret}}
\label{app:horizon-regret-proof}
The corollary follows from Theorem~\ref{thm:directional-delay} and the delay bound \eqref{eq:directional-delay-envelope}. The notation $\widetilde O$ suppresses logarithmic factors in $T$.

\begin{proof}[Proof of Corollary~\ref{cor:horizon-regret}]
On the confidence event of Theorem~\ref{thm:directional-delay}, satisfying \eqref{eq:dual-stop} at every nonempty epoch gives
\[
 R_T\leq(2+\alpha)\omega\overline\beta_T
       \sqrt{N_T(2\Lambda_T+H_T)}.
\]
The delay bound in Appendix~\ref{app:directional-proof} and the workload bound give
\[
 H_T\leq\frac{2MD}{\lambda_0}\Lambda_T,
 \qquad N_T\leq MT.
\]
Substituting both inequalities yields
\begin{equation}
 R_T\leq(2+\alpha)\omega\overline\beta_T
       \sqrt{2MT\left(1+\frac{MD}{\lambda_0}\right)\Lambda_T}.
 \label{eq:horizon-regret}
\end{equation}
For fixed ${D_\phi},M,\lambda_0$, $\Lambda_T\leq {D_\phi}\log(1+MT/(\lambda_0{D_\phi}))=O(\log(1+T))$. For the horizon rate, we also hold $D,\alpha,\omega,\sigma,S,\delta$ fixed, so
\[
 \overline\beta_T
 =\sigma\sqrt{\Lambda_T+2\log(1/\delta)}+\sqrt{\lambda_0}S
 =O(\sqrt{\log(1+T)}).
\]
Equation~\eqref{eq:horizon-regret} and $\overline\beta_T\leq\sigma\sqrt{\Lambda_T}+\sigma\sqrt{2\log(1/\delta)}+\sqrt{\lambda_0}S$, with $\Lambda_T\leq D_\phi\log(1+MT/(\lambda_0D_\phi))$, give the rate in \eqref{eq:horizon-rates}, so $R_T=\widetilde O(\sqrt T)$ with probability at least $1-\delta$. For $\lambda_0=\max\{1,M\}$,
\[
 R_T=\widetilde O\Bigl((2+\alpha)\,\omega\sqrt{1+D}\,
 \bigl(\sigma D_\phi\sqrt{MT}+SM\sqrt{D_\phi T}\bigr)\Bigr).
\]
With $k=M$ items per round and $d=D_\phi$, the two terms mirror the $\widetilde O(d\sqrt{kT})$ and $\widetilde O(Sk\sqrt{dT})$ cases of C$^2$UCB \citep{takemura2021contextual}, and, for $D=0$, the first improves the task-count dependence of CombLinUCB's $\widetilde O(Kd\sqrt{n})$ bound \citep{wen2015comblin} by a factor $\sqrt M$, with $K=M$ and $n=T$. Quality-aware stopping adds the factor $2+\alpha$, and the realized term $H_T$ in Theorem~\ref{thm:directional-delay} can replace the worst-case delay factor.
\end{proof}

\section{Computational Efficiency}\label{app:computation}

\subsection{Proof of Theorem~\ref{thm:compute-regret}}
\label{app:compute-proof}

Let $\nu_t$, $N^{\mathrm{r}}_t$ and $G^{\mathrm{g}}_t$ denote the component count, retained-state count and state-generation work. The initial gap is $B_t=U_{0,t}-L_{0,t}$, where $U_{0,t}$ is the incumbent's optimistic cost and $L_{0,t}$ sums the component minima. The \emph{tolerance lower bound} is
\begin{equation}
 \tau_t=\alpha\omega\beta_t
       \sum_{r\in\cR_t}\min_{g:(r,g)\in\cE_t}
                         \|\boldsymbol{\phi}_{rg,t}\|_{\boldsymbol{V}_t^{-1}}.
 \label{eq:stopping-floor}
\end{equation}
This quantity lower-bounds the stopping threshold of every feasible schedule. After a failed initial test, \textsc{RLS} starts doubling refinement at $\kappa=\nu_t$. Write $\kappa_t$ for the chosen or last completed resolution. Within a call, optimistic costs and confidence widths stay fixed.

\begin{proof}[Proof of Theorem~\ref{thm:compute-regret}]
\par\smallskip\noindent\textbf{Completed solves bound the optimization error.}\quad
Fix a nonempty epoch $t$ and omit its subscript. Lemma~\ref{lem:subset-original} shows that the component representation preserves the optimistic optimum $m^-$. Thus the component minima and input incumbent satisfy
\[
 L_0\leq m^-\leq U_0,\qquad B=U_0-L_0\geq0.
\]
If $B=0$, the input incumbent and $L_0$ give zero optimization error and satisfy \eqref{eq:dual-stop}, even at zero tolerance.

For $B>0$, recurrence~\eqref{eq:rounded-component-dp} keeps the minimum exact reservation for each rounded excess-cost total at resolution $\kappa$. Backtracking therefore gives a feasible decision $\boldsymbol{a}_\kappa$. The lower bound in \eqref{eq:rounded-lower-bound} and Appendix~\ref{app:rounding-proof} give
\[
 L_\kappa\leq m^-\leq J^-(\boldsymbol{a}_\kappa),\qquad
 J^-(\boldsymbol{a}_\kappa)-L_\kappa\leq\frac{\nu B}{\kappa}.
\]
Rounding changes only excess costs. Keeping reservations exact preserves feasibility in both the capacity test and the reconstructed schedule.

\textsc{JAS} retains the best feasible schedule and largest valid global lower bound found so far. Hence $J^-(\boldsymbol{a}^0)\leq J^-(\boldsymbol{a}_\kappa)$ and $L^-\geq L_\kappa$, giving
\[
 0\leq J^-(\boldsymbol{a}^0)-m^-
 \leq J^-(\boldsymbol{a}^0)-L^-
 \leq J^-(\boldsymbol{a}_\kappa)-L_\kappa
 \leq\frac{\nu B}{\kappa}.
\]
The retained schedule and bound may come from different completed solves. Restoring epoch subscripts gives \eqref{eq:component-main-gap}.

\par\smallskip\noindent\textbf{DP work scales with the retained states and resolution.}\quad
Generate states once in $O(G^{\mathrm{g}}_t)$ work, storing each state's minimum-cost assignments, cost, exact reservation and selected confidence-width sum. Recurrence~\eqref{eq:rounded-component-dp} considers at most $\kappa_t+1$ totals per retained state, requiring $O(N^{\mathrm{r}}_t(\kappa_t+1))$ work. Backtracking selects one state per component, and their stored width sums give the stopping test in $O(\nu_t)$ work. Expanding the final assignments is covered by $O(G^{\mathrm{g}}_t)$. Since each component retains at least one state, $\nu_t\leq N^{\mathrm{r}}_t$, so these operations fit the same bound. The per-epoch computational complexity is therefore $O(G^{\mathrm{g}}_t+N^{\mathrm{r}}_t(\kappa_t+1))$. Summing completed solves across epochs gives total computational complexity
\[
 O\!\left(\sum_{t\leq T}
       [G^{\mathrm{g}}_t+N^{\mathrm{r}}_t(\kappa_t+1)]\right).
\]
Here $\kappa_t$ is the chosen or last completed resolution. Under doubling, the geometric-sum argument below includes all earlier completed solves in the same bound, with states generated only once.

\par\smallskip\noindent\textbf{A positive tolerance lower bound ensures finite refinement.}\quad
For a feasible decision $\boldsymbol{a}$, define its selected confidence-width sum as
\[
 w_t(\boldsymbol{a})=\sum_{(r,g):x_{rg}=1}
                    \|\boldsymbol{\phi}_{rg,t}\|_{\boldsymbol{V}_t^{-1}}.
\]
Each task uses exactly one compatible option, so the tolerance lower bound in \eqref{eq:stopping-floor} satisfies
\[
 \alpha\omega\beta_t w_t(\boldsymbol{a})
 \geq \alpha\omega\beta_t
       \sum_{r\in\cR_t}\min_{g:(r,g)\in\cE_t}
                     \|\boldsymbol{\phi}_{rg,t}\|_{\boldsymbol{V}_t^{-1}}
 =\tau_t
\]
for every feasible incumbent throughout refinement. The context matrix, quality estimates, component states, $U_0$, $L_0$ and $B_t$ remain fixed during the call.

If the input incumbent and $L_0$ satisfy \eqref{eq:dual-stop}, no rounded solve is needed. Otherwise, for $\tau_t>0$,
\[
 B_t>
 \alpha\omega\beta_t w_t(\boldsymbol{a}_t^0)\geq\tau_t.
\]
At stage $j=0,1,\ldots$, use $\kappa_j=\nu_t2^j$. The completed-solve bound gives
\[
 J_t^-(\boldsymbol{a}_t^0)-L_t^-
 \leq\frac{\nu_tB_t}{\kappa_j}
 =B_t2^{-j}.
\]
Thus the stopping condition holds no later than stage
\[
 j_* =\left\lceil\log_2\frac{B_t}{\tau_t}\right\rceil,
 \qquad
 B_t2^{-j_*}\leq\tau_t
 \leq\alpha\omega\beta_t w_t(\boldsymbol{a}_t^0).
\]
The incumbent's selected uncertainty may allow earlier stopping. The lower bound $\tau_t$ controls the worst-case number of stages.

\par\smallskip\noindent\textbf{Doubling adds only a constant factor to the last solve's work.}\quad
Since $\kappa_j\geq1$, the sum of the resolution ranges obeys
\begin{align*}
 \sum_{j=0}^{j_*}(\kappa_j+1)
 &\leq 2\sum_{j=0}^{j_*}\kappa_j
   =2\nu_t(2^{j_*+1}-1)
   <4\kappa_{j_*},\\
 \kappa_{j_*}&=\nu_t2^{j_*}
   <2\nu_t\frac{B_t}{\tau_t}.
\end{align*}
Each stage's DP work covers its $O(\nu_t)$ stopping test. Including state generation, refinement after a failed initial test costs $O(G^{\mathrm{g}}_t+N^{\mathrm{r}}_t\nu_tB_t/\tau_t)$. Including successful initial tests and $B_t=0$ gives
\[
 O\!\left(G^{\mathrm{g}}_t+N^{\mathrm{r}}_t\nu_t
                  \left(1+\frac{B_t}{\tau_t}\right)\right),
\]
which proves \eqref{eq:statistical-work}. Empty task epochs require no component solve and contribute zero regret.
\end{proof}

\paragraph{Implications for decision quality.}
If a resolution $\kappa_t$ completes at every nonempty epoch, then on the confidence event of Theorem~\ref{thm:directional-delay}, substituting \eqref{eq:component-main-gap} into \eqref{eq:directional-regret} gives
\begin{equation}
 R_T\leq\sum_{t\leq T}\frac{\nu_tB_t}{\kappa_t}
       +2\omega\overline\beta_T\Xi_T.
 \label{eq:work-regret}
\end{equation}

\paragraph{Deadlines and zero tolerance.}
An interrupted solve keeps the retained feasible schedule and valid global lower bound, but the unfinished resolution gives no new accuracy guarantee. The finite work bound requires $\tau_t>0$. For $\tau_t=0$, stopping can still succeed when the incumbent and lower bound close the gap. Integer reservations also permit \eqref{eq:exact-component-dp}. These bounds count scheduling arithmetic operations, with optimistic costs and confidence widths given. They exclude quality estimation, model execution and feedback collection. A fixed tolerance $\epsilon>0$ uses $\kappa=\max\{1,\lceil\nu B/\epsilon\rceil\}$ and $O(G^{\mathrm{g}}+N^{\mathrm{r}}(1+\nu B/\epsilon))$ work. Quality-aware stopping adapts this accuracy to the remaining quality uncertainty.

\paragraph{Workload and precision.}
A smaller initial gap $B$ or a larger tolerance lower bound $\tau$ reduces the complexity bound. Centering makes the bound invariant to common cost offsets within each component. For one nonempty epoch, omit $t$ and let ${M_{\mathrm{c}}}=\max_i|\cR_i|$. Appendix~\ref{app:component-regime-proof} gives $N^{\mathrm{r}}\leq\nu(2^{K_{\mathrm{c}}}-1)$ and $G^{\mathrm{g}}=O(\nu2^{K_{\mathrm{c}}}{K_{\mathrm{c}}}({M_{\mathrm{c}}}+1))$ for result-subset enumeration, so \eqref{eq:statistical-work} becomes
\[
 O\!\left(2^{K_{\mathrm{c}}}\nu\left[{K_{\mathrm{c}}}({M_{\mathrm{c}}}+1)+\nu\left(1+\frac{B}{\tau}\right)\right]\right).
\]
At a fixed resolution $\kappa$, Corollary~\ref{cor:regime} gives work linear in $\nu$ for bounded ${K_{\mathrm{c}}},{M_{\mathrm{c}}}$. Appendix~\ref{app:subset-proof} describes the different generation cost of the assignment catalogues used in timing experiments.

\paragraph{Comparison with exact scheduling.}
\textsc{RLS} indexes rounded costs, while Exact DP indexes integer capacity. Under Assumption~\ref{ass:components}, completed solves have the following per-epoch complexity bounds.
\begin{center}
\small
\begin{tabular}{@{}ll@{}}
\toprule
Method & Arithmetic work or candidate count \\
\midrule
Complete result-selection enumeration & $2^{\sum_i|\cG_i|}$ candidates \\
Exact integer-capacity DP & $O(G^{\mathrm{g}}+N^{\mathrm{r}}({c^{\max}}+1))$ \\
\textsc{RLS}, to \eqref{eq:dual-stop} ($\tau>0$) & $O(G^{\mathrm{g}}+N^{\mathrm{r}}\nu(1+B/\tau))$ \\
\bottomrule
\end{tabular}
\end{center}
Exact DP uses the recurrence in Appendix~\ref{app:exact-capacity}. After a failed initial test with $\tau>0$, $\kappa=\lceil\nu B/\tau\rceil$ suffices, and doubling reaches the stopping condition at a resolution below $2\nu B/\tau$. Thus \textsc{RLS} can reduce DP work when ${K_{\mathrm{c}}}$ is small and $\nu B/\tau\ll {c^{\max}}$, unless state generation dominates. Small capacities or tight tolerances can favor Exact DP.

\section{Additional Experimental Evaluation}\label{app:online-v9}

\subsection{Experimental Settings}\label{app:experimental-protocol}
\paragraph{Data version and splits.}
Workloads use the published releases of QuALITY, QASPER, MultiRC and RACE cited in Section~\ref{sec:experiments}. The source and target models are Qwen3.5-4B and Qwen3.5-9B, respectively. Each dataset is divided into disjoint development and confirmation pools in a $35$:$80$ ratio. The default workload draws four documents per epoch and schedules their questions as tasks.

\paragraph{Routes.}
We use four routes: native, source-short, source-long and target-self. Native execution uses the $9$B target to answer from the original document. Source-short and source-long use the $4$B source to prepare concise and detailed summaries, respectively, from which the target answers. Target-self uses the $9$B model for both summary preparation and answering. Summary length is five percent of the source document or $16$ tokens, whichever is larger. SMU uses native execution and target-self. The five primary methods use all four routes.

\paragraph{Context construction.}
The context vector uses a scalar $\xi$ from document length and question--document overlap, with two coordinates for each route. For character length $L_{\mathrm{doc}}$ and overlap $o\in[0,1]$, where $o$ is the fraction of question tokens that appear in the document, let $\tilde d$ be the min--max normalization of $\log(1+L_{\mathrm{doc}})$ in the dataset pool, and set
\begin{equation}
 \xi=\tfrac12\min\bigl\{1,\max\{0,\tilde d\}\bigr\}+\tfrac12(1-o).
\end{equation}
For route $k$, the context has $\boldsymbol{\psi}=(1-\xi,\xi)$ in its two coordinates and zeros in the other route blocks. Thus ${D_\phi}=8$ and $\|\boldsymbol{\phi}_{rg,t}\|_2\leq 1$.

\paragraph{Quality feedback and cost calibration.}
Quality losses are computed from the source and target LLM answers: $0/1$ error on QuALITY and RACE, $1-\mathrm{F1}$ on QASPER, and per-question F1 loss on MultiRC. The primary comparison uses delay $D=0$. Prefill and decode times per token are obtained from one calibration of the two models. They determine the modeled costs of native answers, summary preparation and answers using a summary.

\paragraph{Quality-aware cost and regret.}
We construct a semi-synthetic evaluation using the task contexts, computation costs, and resource constraints described above. For each route $k$, we fit $\boldsymbol{\theta}^{\mathrm{sim}}_k \in [0,1]^2$ by constrained least squares on the development-pool LLM losses and freeze it before confirmation experiments. The simulator defines $\mu_{rg,t} =\boldsymbol{\phi}_{rg,t}^{\top} \boldsymbol{\theta}^{\mathrm{sim}}\in[0,1]$. For every executed assignment, it draws a fresh observation $\ell^{\mathrm{o}}_{rg,t}\sim\operatorname{Bernoulli}(\mu_{rg,t})$, independently across executions conditional on the contexts and fixed simulator parameters. Only selected observations are revealed to the learner, with delays independent of unrevealed outcomes.

The evaluator uses $\mu_{rg,t}$ in \eqref{eq:objective} and computes the exact per-epoch optimum over the same feasible set. Regret is the cumulative difference between the chosen schedule's mean cost and this optimum, as in \eqref{eq:online-problem}; the per-task average divides this quantity by $N_T=\sum_{t=1}^{T}|\mathcal{R}_t|$. The simulator parameters and means are unavailable to the learner. This construction satisfies the linear quality model by design; it does not assume that actual LLM losses are exactly linear in the chosen features.

\subsection{Computational Efficiency}\label{app:computation-experiments}
\subsubsection{Stopping Rule and Solver Diagnostics}\label{app:solver-diagnostics}
Table~\ref{tab:v9-ablation} reports ablations of individual \textsc{GLS} mechanisms through $400$ epochs. 

\begin{table}[htbp]\centering\footnotesize
\setlength{\tabcolsep}{2.4pt}
\caption{\textsc{GLS} mechanism ablations.}
\label{tab:v9-ablation}
\begin{tabular}{@{}lcccccc@{}}\toprule
 & QARS & \shortstack{No\\screening} & \shortstack{No collective\\moves} & \shortstack{Fixed\\tolerance} & \shortstack{Zero\\tolerance} & \shortstack{No greedy\\initialization}\\\midrule
Regret & 0.785 & 0.845 & 0.808 & 0.896 & 0.787 & 0.875\\
Time (ms) & 1.770 & 1.786 & 2.551 & 1.469 & 4.291 & 1.337\\
\bottomrule\end{tabular}
\end{table}

Screening removes an average of $0.693$ results per epoch. Collective moves reduce MILP calls from $0.199$ to $0.094$ per epoch. On the primary trajectories, the mean quality-loss interval width is $0.278$, and $98.7\%$ of decisions satisfy the stopping rule \eqref{eq:dual-stop}.

\subsubsection{Resources and Scheduling Deadlines}\label{app:resource-sensitivity}
Figure~\ref{fig:online-v9-budgets} reports regret under capacity, workload-size, and scheduling-deadline sweeps. Capacities are $2,3,4,5,6$ and $8$; the workload-size sweep varies document count and capacity jointly. Deadlines are $1,3,10,30,100$\,ms. Each point uses $400$ epochs, with task streams paired across methods. \alg{} has the lowest mean regret in several settings. Figure~\ref{fig:scheduling-time-deadline} reports scheduling time for the same five methods and deadlines, with $400$ epochs per dataset.

\begin{figure}[htbp]\centering
\includegraphics[width=0.6\linewidth]{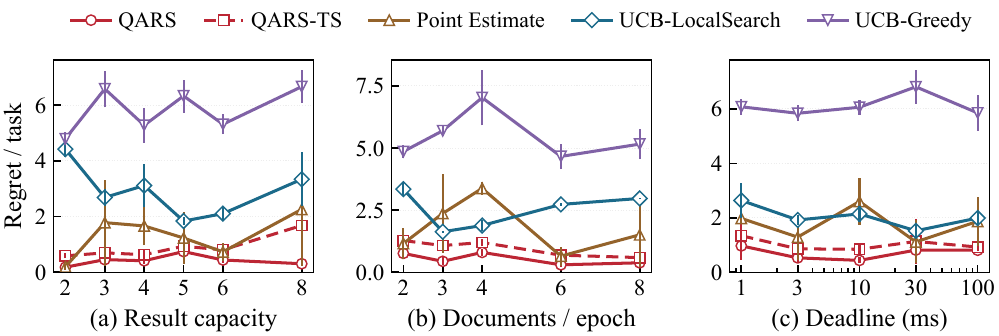}
\vspace{-4pt}
\caption{Regret across resource and deadline settings for the five methods in
Table~\ref{tab:v9-main-primary}.}
\label{fig:online-v9-budgets}
\end{figure}
\vspace{-4pt}

\subsubsection{Component Dispatch and Scaling}\label{app:component-experiments}
Figure~\ref{fig:online-v16-scale} compares Bounded (\textsc{RLS}) and General (\textsc{GLS}) with Exact DP and MILP as the component count $\nu$ grows. All methods use fixed quality estimates and assignment catalogues. The comparison uses integer-component instances and a $50$\,ms scheduling budget, with ten instances at each $\nu\in\{24,32,40,48,64\}$. Reported times include catalogue construction and any budget overrun (Appendix~\ref{app:exact-capacity}); MILP times also include model construction in SciPy/HiGHS. At $\nu=24$, the times are $1.70$\,ms for Bounded, $3.49$\,ms for General, $5.38$\,ms for Exact DP and $9.91$\,ms for MILP.

\subsubsection{Subset versus Assignment Catalogues}\label{app:subset-assignment-exp}
Figure~\ref{fig:subset-assignment} compares subset and assignment catalogue generation (Corollary~\ref{cor:regime}). The scaling experiment above fixes ${K_{\mathrm{c}}}={M_{\mathrm{c}}}=4$. Here $\nu=24$ and quality estimates are fixed. We vary ${M_{\mathrm{c}}}$ with ${K_{\mathrm{c}}}=4$, then ${K_{\mathrm{c}}}$ with ${M_{\mathrm{c}}}=4$. Times are medians over ten instances.

\begin{figure}[ht]\centering
\includegraphics[width=0.6\linewidth]{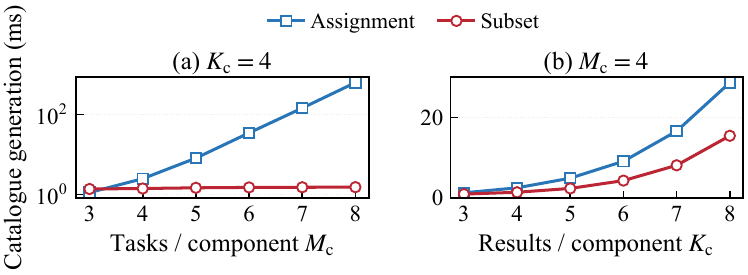}
\vspace{-4pt}
\caption{Catalogue generation time. (a)~${K_{\mathrm{c}}}=4$, log scale.
(b)~${M_{\mathrm{c}}}=4$, linear scale.}
\label{fig:subset-assignment}
\end{figure}

At ${M_{\mathrm{c}}}=8$, assignment generation takes $629$\,ms and subset generation takes $1.53$\,ms. At ${K_{\mathrm{c}}}=8$, the times are $28.6$\,ms and $15.4$\,ms. Both catalogues yield the same objective on all $120$ instances. The subsequent dynamic program takes about $0.2$\,ms and reconstruction about $0.01$\,ms.
\subsection{Quality Weight}\label{app:omega-sweep}
Figure~\ref{fig:online-v9-omega} multiplies each dataset's default quality weight $\omega_0$ by $0.25$, $0.5$, $1$, $2$ and $4$. The five methods in Table~\ref{tab:v9-main-primary} are evaluated at capacity $5$ for $400$ epochs. At $\omega=\omega_0$, \alg{} has the lowest mean regret ($0.863$). The factor $\omega$ in Theorem~\ref{thm:directional-delay} reflects the weighting of quality loss in the objective.
\begin{figure}[htbp]\centering
\includegraphics[width=0.8\linewidth]{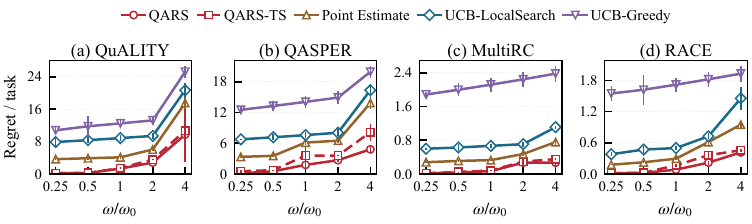}
\vspace{-10pt}
\caption{Regret versus quality-weight multiplier. $\omega_0$ is the dataset-specific default.}
\label{fig:online-v9-omega}
\end{figure}
\vspace{-8pt}

\end{document}